\documentclass[final,leqno]{amsart}
\usepackage{amsmath,amssymb}

\usepackage{graphicx}
\usepackage{caption,subcaption}
\usepackage{tensor}
\usepackage[norelsize,boxed,ruled]{algorithm2e}
\usepackage{todonotes}

\newcommand{\pder}[2]{\ensuremath{\frac{ \partial #1}{\partial #2}}}

\newtheorem{theorem}{Theorem}[section]
\newtheorem{definition}[theorem]{Definition}
\newtheorem{proposition}[theorem]{Proposition}
\newtheorem{corollary}[theorem]{Corollary}
\newtheorem{lemma}[theorem]{Lemma}
\newtheorem{example}[theorem]{Example}
\newtheorem{ass}[theorem]{Assumption}

\DeclareMathOperator{\Diff}{Diff}
\DeclareMathOperator{\Dens}{Dens}

\DeclareMathOperator{\OdeSolve}{OdeSolve}
\DeclareMathOperator{\Tr}{Tr}

\title{Qualitatively accurate spectral schemes for advection and transport}
\author{Henry O. Jacobs \& Ram Vasudevan}
\date{\today}

\begin{document}

\maketitle

\begin{abstract}
	The transport and continuum equations exhibit a number of conservation laws.
	For example, scalar multiplication is conserved by the transport equation, while
	positivity of probabilities is conserved by the continuum equation.
	Certain discretization techniques, such as particle based methods, conserve these properties, but converge slower than spectral discretization methods on smooth data.
	Standard spectral discretization methods, on the other hand, do not conserve the invariants of the transport equation and the continuum equation.
	This article constructs a novel spectral discretization technique that conserves these important invariants while simultaneously preserving spectral convergence rates. 
	The performance of this proposed method is illustrated on several numerical experiments. 
\end{abstract}

\section{Introduction}
\label{sec:intro}

Let $M$ be a compact $n$-manifold with local coordinates $x^{1},\dots,x^{n}$
and let $X$ be a smooth vector-field on $M$, whose local components are given by $(X^{1}(x), \dots, X^{n}(x) )$. 
This paper is concerned with the following pair of partial differential equations (PDEs):
\begin{align}
	\partial_{t} f + X^{i} \partial_{i} f = 0 \label{eq:function pde} \\
	\partial_{t} \rho + \partial_{i} (\rho X^{i}) = 0 \label{eq:density pde}
\end{align}
for a time dependent function $f$ and a time dependent density $\rho$ on $M$.
In the above PDEs we are following the Einstein summation convention, and summing over the index ``$i$.''
Equation \eqref{eq:function pde}, which is sometimes called the ``transport equation,'' describes how a scalar quantity is transported by the flow of $X$~\cite{Truesdell1991}.
Equation \eqref{eq:density pde}, which is sometimes called the ``continuum equation'' or ``Liouville's equation'' describes how a density (e.g. a probability distribution) is transported by the flow of $X$.
Such PDEs arises in a variety of contexts, ranging from  mechanics~\cite{Batchelor1999,Truesdell1991} to control theory~\cite{HenrionKorda2014}, and can be seen as zero-noise limits of the forward and backward Kolmogorov equations~\cite{Oksendal2003}.

The solution to \eqref{eq:function pde} takes the form $f(x;t) = f( (\Phi_{X}^{t})^{-1}(x) ; 0) \equiv (\Phi_{X}^{t})_{*} f(\cdot;0)$
where $\Phi_{X}^{t}:M \to M$ is the flow map of $X$ at time $t$ ~\cite[Chapter 18]{Lee2006}.
From this observe that \eqref{eq:function pde} exhibits a variety of conservation laws.
For example, if $f$ and $g$ are solutions to \eqref{eq:function pde},
then so is their product, $f \cdot g$, and their sum, $f+g$.
Similarly, the solution to \eqref{eq:density pde} takes the form $\rho(x;t) = | \det( D(\Phi_{X}^{t})^{-1}(x) ) | \rho( (\Phi_{X}^{t})^{-1}(x) ;0) := (\Phi_{X}^{t})_{*} \rho(\cdot;0)$.
One can deduce that the $L^{1}$-norm of $\rho(x;t)$ is conserved in time~\cite[Theorem 16.42]{Lee2006}.
Finally, \eqref{eq:density pde} is the adjoint evolution equation to \eqref{eq:function pde} in the sense that the integral $\langle f , \rho \rangle := \int f \rho$ is constant in time\footnote{To see this compute
$\frac{d}{dt} \langle f , \rho \rangle = \int ( \partial_{t} f ) \rho + f ( \partial_{t} \rho) $.  One finds that the final integral vanishes upon substitution of \eqref{eq:function pde} and \eqref{eq:density pde} and applying integration by parts.}.
This motivates the following definition of qualitative accuracy:

\begin{definition} \label{def:quality}
	    A numerical method for \eqref{eq:function pde} and \eqref{eq:density pde} is \emph{qualitatively accurate} if it conserves discrete analogs of scalar multiplication/addition, the $L^{1}$-norm and the total mass for densities and the sup-norm for functions.
\end{definition}

Both \eqref{eq:function pde} and \eqref{eq:density pde} can be numerically solved by a variety of schemes.
For a continuous initial condition, $f_0$, for example, the method of characteristics ~\cite{Evans2010,MTA} describes a solution to \eqref{eq:function pde} as a time-dependent function $f( x_{t} ;t) = f_{0}( x_{0} )$ where $x_{t}$ is the solution to $\dot{x} = X(x)$.
This suggests using a particle method to solve for $f$ at a discrete set of points~\cite{Leveque1992}.
 In fact, a particle method would inherit many discrete analogs of the conservation laws of \eqref{eq:function pde}, and would as a result be \emph{qualitatively accurate}.
For example, given the input $h_{0} = f_{0} \cdot g_{0}$, the output of a particle method is identical to the (componentwise) product of the outputs obtained from inputing $f_{0}$ and $g_{0}$ separately.
However, particle methods converge much slower than their spectral counterparts when the function $f$ is highly differentiable~\cite{Boyd2001,Gottlieb2001}.

In the case where $M$ is the unit circle, $S^{1}$, a spectral method can be obtained by converting \eqref{eq:density pde} to the Fourier domain where it takes the form
of an Ordinary Differential Equation (ODE):
$
	\frac{d \hat{\rho}_{k}}{dt} + 2\pi i k  \widehat{X}_{k-\ell} \hat{f}_{\ell} 
$
where $\hat{\rho}_{k}$ and $\widehat{X}_{k}$ denote the Fourier transforms of $\rho$ and $X$ ~\cite{Taylor1974}.
In particular, this transformation converts \eqref{eq:density pde} into an ODE on the space of Fourier coefficients.
A standard spectral Galerkin discretization is obtained by series truncation. 

Such a numerical method is good for $C^{k}$-data, in the sense that the convergence rate, over a fixed finite time $T>0$, is faster than $\mathcal{O}(N^{-k})$  where $N$ is the order of truncation~\cite{Boyd2001,Gottlieb2001,Gottlieb1977numerical}.
In particular, spectral schemes converge faster than particle methods when the initial conditions have some degree of regularity.
Unfortunately the spectral algorithm given above is not \emph{qualitatively accurate}, as is demonstrated by several examples in Section \ref{sec:numerics}.

The goal of this paper is
\emph{to find a numerical algorithm for \eqref{eq:function pde} and \eqref{eq:density pde} which is simultaneously stable, spectrally convergent, and qualitatively accurate.}

\subsection{Previous work}
Within mechanics, spectral methods for the continuum and transport equation are a common-place where they are viewed as special cases of first order hyperbolic PDEs~\cite{Boyd2001,Gottlieb2001}.
Various Galerkin discretizations of the Koopman operator\footnote{The Koopman operator is a linear operator ``$K$'' which yields the solution $f = K \cdot f_{0}$ to \eqref{eq:function pde}. We refer the reader to~\cite{BudisicMohrMezic2012} for a survey of recent applications.} have been successfully used for generic dynamic systems~\cite{BudisicMohrMezic2012,Mezic2005}, most notably fluid systems~\cite{Rowley2009} where such discretizations serve as a generalization of dynamic mode decomposition~\cite{Schmid2010}.
Dually, Ulam-type discretizations of the Frobenius-Perron operator~\cite{LasotaMackey1994,Ulam1947} have been used to find invariant manifolds of systems with uniform Gaussian noise~\cite{FroylandJungeKoltai2013,FroylandPadberg2009}.
In continuous time, Petrov-Galerkin discretization of the infinitesimal generator of the Frobenius Perron operator converge in the presence of noise ~\cite{BittracherKoltaiJunge2015} and preserve positivity in a Haar basis~\cite{koltai2011thesis}.

In this article, we consider a unitary representation of the diffeomorphisms of $M$ known to representation theorists~\cite{Ismagilov1975,VershilGelfandGraev1975} and quantum probability theorists~\cite{Meyer1998,Parthasarathy2012}.
To be more specific, we consider the action of diffeomorphisms on the Hilbert space of half-densities~\cite{BatesWeinstein1997,GuilleminSternberg1970}.
Half densities can be abstractly summarized as an object whose ``square'' is a density
or, alternatively, can be understood as a mathematician's nomenclature for a physicist's ``wave functions.''
One of the benefits of working with half-densities, over probability densities, is that the space of half-densities is a Hilbert space, while the space of probability densities is a convex cone~\cite{GuilleminSternberg1970}.
This tactic of inventing the square-root of an abstract object in order to simplify a problem has been used throughout mathematics.
The most familar example would be the invention of the complex numbers to find the roots of polynomials~\cite{Stewart2015}.
A more modern example within applied mathematics can be found in~\cite{Balci2011} where the (conic) space of positive semi-definite tensor fields which occur in non-Newtonian fluids is transformed into the (vector) space of symmetric tensors~\cite{Balci2011}.
Similarly, an alternative notion of half-densities is invoked in~\cite{Crane2013} to transform the mean-curvature flow PDE into a better behaved one.

\subsection{Main contributions}

In this paper we develop numerical schemes for \eqref{eq:function pde} and \eqref{eq:density pde}.
First, we derive an auxiliary PDE, \eqref{eq:half density pde}, on the space of half-densities in Section \ref{sec:half densities}.
We relate solutions of \eqref{eq:half density pde} to solutions of \eqref{eq:function pde} and \eqref{eq:density pde} in Theorem \ref{thm:quantize}.
Second, we pose an auxiliary spectral scheme for \eqref{eq:half density pde} in Section \ref{sec:discretization}.
Our auxiliary scheme induces numerical schemes for \eqref{eq:function pde} and \eqref{eq:density pde} via Theorem \ref{thm:quantize}.
Third, we derive a spectral convergence rate for our auxiliary scheme in Section \ref{sec:analysis}.
The spectral convergence rate for our auxiliary scheme induce spectral convergence rates for numerical schemes for \eqref{eq:function pde} and \eqref{eq:density pde}.
Finally, we prove our schemes are qualitatively accurate, as in Definition \ref{def:quality}, in Section \ref{sec:qualitative}.
We end the paper by demonstrating these findings in numerical experiments in Section \ref{sec:numerics}.
We observe our algorithm for \eqref{eq:density pde} to be superior to a standard spectral Galerkin discretization, both in terms of numerical accuracy and qualitative accuracy.

\subsection{Notation}
Throughout the paper $M$ denotes a smooth compact $n$-manifold without boundary. 
The space of continuous complex valued functions is denoted $C(M)$ and has a topology induced by the sup-norm, $\| \cdot \|_{\infty}$ (see~\cite{Taylor1974,Rudin1991,MTA,Conway1990}).
Given a Riemannian metric, $g$, the resulting Sobolev spaces on $M$ are denoted $W^{k,p}(M ; g)$ (see~\cite{Hebey1999}).
The tangent bundle to $M$ is denoted by $TM$, and the $n$th iterated Whitney sum is denoted by $\bigoplus^{n} TM$ (see~\cite{Lee2006,MTA}).
A (complex) density is viewed as a continuous map $\rho: \bigoplus^{n}TM \to \mathbb{C}$ which satisfies certain geometric properties which permit a notion of integration.
We denote the space of densities by $\Dens(M)$ and the integral of $\rho \in \Dens(M)$ is denoted by $\int \rho$~\cite[Chapter 16]{Lee2006}.
By completion of $\Dens(M)$ with respect to the norm $\| \rho \|_{1} := \int | \rho|$ we obtain a Banach space, $L^{1}(M)$.
We should note that $L^1(M)$ is homeomorphic to the space of distributions up to choosing a partition of unity of $M$.
Given a function $f \in C(M)$, we denote the multiplication of $\rho \in L^{1}(M)$ by $f \rho$, and we denote the dual-pairing by $\langle f , \rho \rangle := \int f \rho$.
We let $W^{s,1}$ denote the closed subspace of $L^{1}(M)$ whose elements exhibit $s>0$ weak derivative~\cite{Hormander2003}.

Given a separable Hilbert space $\mathcal{H}$ we denote the Banach-algebra of bounded operators by $B( \mathcal{H})$ and topological group of unitary operators by $U( \mathcal{H})$.
The adjoint of an operator $L : \mathcal{H} \to \mathcal{H}$ is denoted by $L^{\dagger}$.
The trace of a trace class operator, $L$, is denoted by $\Tr(L)$.
The commutator bracket for operators $A,B$ on $\mathcal{H}$ is denoted by $[A,B] := A \cdot B - B \cdot A$ (see~\cite{Conway1990}).

\section{Insights from Operator Theory}
\label{sec:operator theory}
Before we describe our algorithms, we take a moment to reflect on the virtue of pursuing qualitative accuracy.
If one knows that some entity is conserved under evolution, then one can reduce the search for solution by scanning a smaller space of possibilities.
For some, this might be justification enough to proceed as we are.
However, more can be said in this case.
The properties that are preserved have a special relationship to \eqref{eq:function pde} and \eqref{eq:density pde}.
It is a result known to algebraic geometers, at least implicitly, that algebra-preservation characterizes \eqref{eq:function pde} completely without any extra ``baggage."
In other words, \emph{the only linear evolution PDE that preserves the algebra of $C(M)$ is \eqref{eq:function pde}.}
As a corollary, the only linear evolution PDE on densities which preserves duality with functions is of the form \eqref{eq:density pde}.
Because of this fact, many nice properties held by \eqref{eq:function pde} and \eqref{eq:density pde} (such as bounds) are also be held by a qualitatively accurate integration scheme.
Therefore, qualitatively accurate schemes leverage the defining aspects of the \eqref{eq:function pde} and \eqref{eq:density pde} to produce numerical approximations with the same qualitative characteristics.

In the remainder of this section, we illustrate how \eqref{eq:function pde} is the unique PDE which preserves $C(M)$ as an algebra.
In the interest of space, we provide references in place of proofs.
To begin, recall the following definitions:
\begin{definition}[\cite{Conway1990}] \label{def:algebra}
	A \emph{Banach-algebra} is a Banach space, $A$, which is equipped with a multiplication-like operation, ``$(a,b) \in A \times A \mapsto ab \in A$,''  that is bounded, ``$\| ab \| \leq \|a \| \|b \|$,'' and associative ``$(ab)c = a(bc)$ for any $a,b,c \in A$.''
	
	A \emph{$C^{*}$-algebra} is a Banach algebra, $A$, over the field $\mathbb{C}$ with an involution, ``$a \in A \mapsto a^{*} \in A$,'' that satisfies:
	\begin{align}
		(ab)^{*} = b^{*} a^{*} \quad,\quad
		(\lambda a)^{*} = \bar{\lambda} a^{*} \quad,\quad
		\| a \| = \| a^{*}\|,
	\end{align}
	for all $a,b \in A$ and $\lambda \in \mathbb{C}$.
	$A$ is \emph{unital} if $A$ has a multiplicative identity.
	$A$ is \emph{commutative} if $ab = ba$ for all $a,b \in A$.
	
	Finally, a map $T:A \to A$ is called a $*$-automorphism if $T$ is a bounded linear automorphism that preserves products, i.e. $T(ab) = T(a) T(b)$.
	We denote the space of $*$-automorphisms of $A$ by $\operatorname{Aut}(A)$.
\end{definition}

The notion of a $C^{*}$-algebra may appear abstract so we provide two important examples:
\begin{example} \label{ex:function algebra}
	Let $X$ be a topological space.
	The space of complex valued continuous functions with compact support, $C_{0}(X)$, is a commutative $C^{*}$-algebra under the sup-norm and the standard addition/multiplication/conjugation operations of complex valued functions.
	If $X$ is compact, then $C_{0}(X) \equiv C(X)$ is unital because the constant function, ``$f(x) = 1$'' is a multiplicative identity.
\end{example}

\begin{example} \label{ex:nc algebra} 
	For a Hilbert space, $\mathcal{H}$, the space of bounded operators, $B(\mathcal{H})$, is a (non-commutative) $C^{*}$-algebra	under operator multiplication and addition, with the involution given by the adjoint mapping ``$L \mapsto L^{\dagger}$'', and the norm given by the operator-norm.
\end{example}

While the notion of a general $C^{*}$-algebra is, \textit{a priori}, more abstract than the examples above, this feeling of abstraction is an illusion.
One of the cornerstones of operator theory is that \emph{all} $C^{*}$-algebras are contained within these examples:

\begin{theorem}[Theorem 1~\cite{GelfandNaimark1943}] \label{thm:GN1}
	Any $C^{*}$-algebra is isomorphic to a sub-algebra of $B(\mathcal{H})$ for some Hilbert space, $\mathcal{H}$.
\end{theorem}

For commutative $C^{*}$-algebras a stronger result holds if one considers the space of character:
\begin{definition}[Space of Characters~\cite{Bondia2001}]
	A \emph{character} of a $C^{*}$-algebra, $A$, is an element of the dual space, $x \in A^{*}$, such that $x(ab) = x(a) x(b)$ for all $a,b \in A$.  
	We denote the space of characters by $X_{A}$.
	For each $a \in A$ there is a function $\hat{a}: X_{A} \to \mathbb{C}$ given by $\hat{a}(x) = x(a)$.
	The map $a \in A \mapsto \hat{a} \in C(X_{A})$ is called the \emph{Gelfand Transform}.
\end{definition}

\begin{proposition}[Lemma 1.1 ~\cite{Bondia2001}]
	$X_{A} \subset A^{*}$ is a compact Hausdorff space with respect to the relative topology if we impost the weak topology on $A^{*}$.
\end{proposition}

For example, if $A$ is a space of continuous complex functions on a manifold $M$, then $X_{A}$ is the space of Dirac-delta distributions, which is homeomorphic to $M$ itself.
The following is a Corollary to \ref{thm:GN1}:

\begin{theorem}[Lemma 1~\cite{GelfandNaimark1943}] \label{thm:GN2}
	Any commutative $C^{*}$-algebra, $A$, is canonically homeomorphic to $C(X_{A} )$.
\end{theorem}

The result of Theorem \ref{thm:GN2} is that all commutative $C^{*}$-algebras are effectively represented by Example \ref{ex:function algebra}.
Historically, Theorems \ref{thm:GN1} and \ref{thm:GN2} have been valued because they turn abstract $C^{*}$-algebras
(as described by the Definition \ref{def:algebra}) into Examples \ref{ex:function algebra} and \ref{ex:nc algebra}.
In this paper, we go the opposite direction.
We start with an evolution equation, \eqref{eq:function pde}, on the space $C(M)$ for a compact manifold $M$ and we transform it into an equation on a commutative sub-algebra of $B(\mathcal{H})$
for a suitably chosen Hilbert space, $\mathcal{H}$ by finding an embedding from $C(M)$ into $B(\mathcal{H})$.
That we seemingly transform ``concrete'' objects into ``abstract'' objects is one possible reason that the algorithms in this paper were not constructed earlier.
However, ``abstract'' does not necessarily imply difficult, with respect to numerics.
In fact, as this paper shows, it is easier to represent advection equations in this operator-theoretic form.
In essence, this is related to a corollary to Theorem \ref{thm:GN2}:

\begin{corollary}[follows from Corollary 1.7 of ~\cite{Bondia2001}]
	Let $M$ be a manifold.
	If $T: C_{0}(M) \to C_{0}(M)$ is $*$-automorphism then there is a unique homeomorphism $\Phi_{T} \in \Diff_{0}(M)$ such that $T[a](x) = a( \Phi_{T}(x) )$.
	That is, $\Diff_{0}(X_{A}) \equiv \operatorname{Aut}(A)$ as a topological group.
	Moreover, a linear evolution equation on $C^{0}(M)$ given by $\partial_{t} f + D[f] = 0$ for some differential operator, $D$,
	preserves  the algebra of $C(M)$ if and only if $D[f] = X^{i} \pder{f}{x^{i}}$ for some vector-field $X$.
	The dual operator is then necessarily of the form $D^{*}[\rho] = \pder{}{x^{i}}( \rho X^{i})$.
\end{corollary}

Said more plainly, \eqref{eq:function pde} and \eqref{eq:density pde} are the generators of all algebra-preserving automorphisms of $C(M)$.
Thus, conservation of sums of products is more than just a fundamental property of \eqref{eq:function pde} and \eqref{eq:density pde}.
\emph{Conservation of sums and products is the defining property of  \eqref{eq:function pde} and \eqref{eq:density pde}.}
As a result, it is natural for this to be reflected in a discretization\footnote{Note: By aiming for qualitative accuracy without sacrificing spectral convergence, we reduce the coefficient of convergence.
Therefore this pursuit makes sense from the standpoint of numerical accuracy as well.}.


\section{Half densities and other spaces}
\label{sec:half densities}

At the core of any Galerkin scheme, including spectral Galerkin, is the use of a Hilbert space upon which everything can be approximated via least squares projections.
The methods we present are no exception.
In this section, we define a canonical $L^{2}$-space associated to a compact manifold $M$, denoted by $L^{2}(M)$ for later use in a spectral discretization\footnote{
We urge the reader familiar with the space $L^{2}_{\mu}(M)$, with respect to some measure $\mu$, not read this section nonetheless.
The $L^{2}$-space we use is slightly different, and this fact permeates the entire article.}.
We also define the Sobolev spaces $H^{s}(M ; g)$ which arise from equipping $M$ with a Riemannian metric, $g$.

For a smooth compact $n$-manifold, $M$, let $\Dens(M)$ denote the space of smooth densities, which we view as anti-symmetric multilinear functions on $\bigoplus^n TM$.
\begin{definition}\label{def:half density}
	A half-density is a smooth complex-valued function $\psi : \bigoplus^n TM \to \mathbb{C}$
	such that $| \psi |^{2} \in \Dens(M)$.
	The space of half densities is denoted by $\sqrt{\Dens(M)}$.
\end{definition}

The following proposition immediately follows from this definition.
\begin{proposition}[see Appendix A ~\cite{BatesWeinstein1997}] \label{prop:half densities}
	If $\psi_{1}, \psi_{2} \in \sqrt{\Dens(M)}$ then the scalar product $\psi_{1} \cdot \psi_{2} : \bigoplus^{n} TM \to \mathbb{C}$ is a complex valued density. 
\end{proposition}

This definition is an equivalent reformulation of the half densities defined in the context of geometric quantization (see~\cite[Chapter 4]{GuilleminSternberg1970} or ~\cite[Appendix A]{BatesWeinstein1997}).
In physical terms, half densities are a geometric manifestation of the wave functions used in quantum mechanics.
It is unfortunate that physicists call these ``wave-functions'' given that they are \emph{not} functions.
To test this assertion, observe how elements of $\sqrt{\Dens(M)}$ transform.
Under a $C^{1}$-automorphism, $\Phi: M \to M$, a half density $\psi \in \sqrt{ \Dens(M) }$ transform to a new half density $\Phi_{*}\psi$ according to the formula
\begin{align}
	(\Phi_{*}\psi)_{x}(v_{1},\dots,v_{n}) := \psi_{ \Phi^{-1}(x) } ( D\Phi^{-1}(x) \cdot v_{1} , \dots , D\Phi^{-1}(x) \cdot v_{n} ) \label{eq:transformation law}
\end{align}
for any $x \in M$ and any $v_{1}, \dots , v_{n} \in T_{x}M$.
This transformation law is inferred by substituting the transformation law for a density into the definition of a half-density.
In other words, this is the unique transformation law such that squaring both sides yields the transformation law for a density.
Notably, this is in contrast to the transformation law for functions, which sends $f \in C^{1}(M)$ to the function $f \circ \Phi^{-1}$.

In local coordinates, $x^{1},\dots,x^{n}$, on an open set $U \subset M$, it is common to write a (complex) density $\rho: \bigoplus^{n}TM \to \mathbb{C}$ as function ``$\rho(x)$'' for $x \in U \subset M$.
This convention is permissible as long as one realizes that what is really meant is that $\rho(x) = f_{\rho}(x) | dx^{1} \wedge \cdots \wedge dx^{n} |$ for some complex valued function $f_{\rho}:U \to \mathbb{C}$.
Therefore, when one writes  ``$\rho(x)$ transforms like $\rho( \Phi^{-1}(x)) \left| \det( D \Phi^{-1}(x) ) \right|$'', what they are really describing is how $f_{\rho}$ is transformed.
The same notational convention can be used to represent half-densities locally as ``functions'' with a different transformation law.
In this case the transformation law for half-densities is locally given by:
\begin{align}
	\tilde{\psi}(x) = \left| \det \left( D\Phi^{-1} (x) \right) \right|^{1/2} \psi \left( \Phi^{-1}(x) \right). \label{eq:local transformation law}
\end{align}

As $|\psi|^{2} \in \Dens(M)$ for any $\psi \in \sqrt{\Dens(M)}$, $|\psi|^{2}$ can be integrated and we observe that half densities are naturally equipped with the norm: $$\| \psi \|_2 :=  \left( \int_M |\psi|^2 \right)^{1/2}$$ which we call the \emph{$2$-norm}.
\begin{definition}
	$L^{2}(M)$ is defined as the completion of $\sqrt{ \Dens(M)}$ with respect to the $2$-norm.
	The space $L^{2}(M)$ is equipped with a complex inner-product given by
	\begin{align}
		\langle \psi \mid \phi \rangle = \int_{M} \bar \psi \phi
	\end{align}
	through polar decomposition, and so $L^{2}(M)$ is a Hilbert space.
\end{definition}

Lastly, given the transformation law for half-densities, \eqref{eq:transformation law} and \eqref{eq:local transformation law}, one can describe how half-densities are transported by the flow, $\Phi_{X}^{t}$, of the vector field, $X$. 
The Lie derivative of a half-density with respect to $X$ is defined as $\pounds_{X}[\psi] = - \left. \frac{d}{dt} \right|_{t=0} (\Phi_{X}^{t})_{*} \psi$ and is given in local coordinates by:
\begin{align}
	\pounds_{X}[\psi] = \frac{1}{2} X^{i} \pder{\psi}{x^{i}} + \frac{1}{2} \pder{}{x^{i}} \left( \psi X^{i} \right). \label{eq:representation}
\end{align}
The advection equation can then be written as:
\begin{align}
	\partial_{t} \psi + \pounds_{X}[\psi] = 0. \label{eq:half density pde}
\end{align}
Despite the Lie derivative being unbounded, a unique solution is defined for all time:
\begin{proposition}[Stone's Theorem ~\cite{Conway1990,Rudin1991}] \label{prop:stone}
	The unique solution to \eqref{eq:half density pde} is of the form $\psi(t) = U(t) \cdot \psi(0)$ where $U(t)$ is the one-parameter semigroup generated by the operator $\pounds_{X}$.
	Explicitly, $U(t)$ is the operator ``$(\Phi_{X}^{t})_{*}$'' in the sense that the solution to \eqref{eq:half density pde} is $\psi(t) = (\Phi_{X}^{t})_{*} \psi(0)$ where $\Phi_{X}^{t}$ is the time flow map of $X$ at time $t$.
\end{proposition}
\begin{proof}
	By inspection we can observe that
	\begin{align*}
		 (\Phi_{X}^{t})_{*} ( \bar{\psi}_{1} \psi_{2}) ) = \overline{ (\Phi_{X}^{t})_{*} \psi_{1} } (\Phi_{X}^{t})_{*} \psi_{2}.
	\end{align*}
	By proposition \ref{prop:half densities}, $\bar{\psi}_{1} \psi_{2}$ is a density, and so we can integrate it.
	The integral of a density is invariant under $C^{1}$ transformations ~\cite[Proposition 16.42]{Lee2006} and we find
	\begin{align*}
		0 = \left. \frac{d}{dt} \right|_{t=0} \int_{M} (\Phi_{X}^{t})_{*} (\bar{\psi}_{1} \psi_{2}) = \left. \frac{d}{dt} \right|_{t=0} \int_{M} \left( \pounds_{X}[\bar{\psi}_{1} ] \psi_{2})  + \bar{\psi}_{1} \pounds_{X}[\psi_{2}] \right) \\
		= \langle \pounds_{X}[\psi_{1}]  \mid \psi_{2} \rangle + \langle \psi_{1} \mid \pounds_{X}[\psi_{2}] \rangle.
	\end{align*}
	Therefore, the operator, $\pounds_{X}$ is anti-Hermitian.
	We can see that $\pounds_{X}$ is densely defined, as it is well defined on $\sqrt{\Dens(M)}$, which is dense in $L^{2}(M)$ by construction.
	Stone's theorem implies that there is a one-to-one correspondence between densely defined anti-Hermitian operators on $L^{2}(M)$
	and one-parameter groups $U(t)$ consisting of unitary operators on $L^{2}(M)$.
	Observe that $\psi(t) = (\Phi_{X}^{t})_{*} \psi(0)$ solve \eqref{eq:half density pde} directly, by taking its time-derivative.
	Thus $U(t) =  (\Phi_{X}^{t})_{*} $ is the unique one-parameter subgroup we are looking for.
\end{proof}

\subsection{The relationship with classical $L^{2}$ spaces}
\label{sec:classical_Lebesgue}
To understand the relationship to classical Lebesgue spaces, recall that for any manifold $M$ (possibly non-orientable) one can assert the existence of a smooth non-negative reference density $\mu$~\cite[Chapter 16]{Lee2006}.
Upon choosing such a $\mu \in \Dens(M)$, the $2$-norm of a continuous complex function $f:M \to \mathbb{C}$ with respect to $\mu$ is
\begin{align}
	\| f \|_{\mu,2} =  \left( \int_M |f|^2 \mu \right)^{1/2}.
\end{align}
and $L^2(M ; \mu)$ is the completion of the space of continuous functions with respect to this norm.
The relationship between $L^{2}(M)$ and $L^{2}(M;\mu)$ is that they are equivalent as topological vector-spaces:
\begin{proposition} \label{prop:non canonical}
	Choose a non-vanishing positive density $\mu : \bigoplus^{n}TM \to \mathbb{R}^{+}$.
	Let $\sqrt{\mu}$ denote the square root of $\mu$\footnote{Explicitly, if $\sqrt{\cdot}: \mathbb{R}^{+} \to \mathbb{R}^{+}$ is the standard square-root function.
	Then $\sqrt{\mu} := \sqrt{\cdot} \circ \mu : \bigoplus^{n} TM \to \mathbb{R}^{+} \subset \mathbb{C}$ is the half-density which we are considering.}.
	For any $\psi \in L^2(M)$ there exists a unique $f \in L^2(M ; \mu)$ such that $\psi = f \cdot  \sqrt{\mu}$.
	This yields an isometry between $L^2(M)$ and $L^2(M ; \mu)$.
\end{proposition}
\begin{proof}
	It suffices to prove that $\sqrt{\Dens(M)}$ is isomorphic to the space of square integrable (w.r.t. $\mu$) continuous functions, because
	the later space is dense in $L^{2}(M;\mu)$.
	Let $\psi \in \sqrt{\Dens(M)}$.  Then $\psi^2$ is a continuous density and there exists a unique function $g \in C^{0}(M)$ such that $\psi^{2} = g \cdot \mu$.
	By taking the square root of both sides we can obtain a unique function $f \in C^{0}$ such that $\psi = f\, \sqrt{\mu}$.
	The function $f$ is unique with respect to $\psi$ and
	the map $\psi \in \sqrt{\Dens(M)} \mapsto f \in C^0(M ; \mathbb{C} )$ sends $\| \cdot \|_{2}$ to $\| \cdot \|_{\mu,2}$ by construction.
	Thus the map is continuous.
	The inverse of the map is given by $f \in C^{0}(M;\mathbb{C}) \mapsto f \, \sqrt{\mu} \in \sqrt{\Dens(M)}$.
\end{proof}

If the spaces are nearly identical the reader may wonder why $L^2(M)$ matters.
In fact, the pair are not identical in all aspects.
As described earlier, under change of coordinates or advection, the elements of each space transform differently.
More importantly, $L^{2}(M)$ is \emph{not} canonically contained within the space of square integrable functions, and functions and densities are \emph{not} contained in $L^{2}(M)$.
Such an embedding may only be obtained by choosing a non-canonical ``reference density'', as in Proposition \ref{prop:non canonical}.
This has numerous consequences in terms of what we can and can not do.
For example, an operator with domain on $L^{2}(M)$ can not generally be applied to objects in $L^{1}(M)$ in the same way.
These limitations can be helpful, since they permit vector fields to act differently on objects in $L^{1}(M)$ than on objects in $L^{2}(M)$.
These prohibitions serve as safety mechanisms, analogous to the use of overloaded functions in object oriented programs, which due to their argument type distinctions, effectively banish certain bugs from arising.


\subsection{Sobolev spaces}
\label{sec:Sobolev spaces}


While the ``canonicalism" of $L^{2}(M)$ is useful for this discussion, the \emph{canonical} Sobolev spaces are not.
Since the algorithms proposed in this paper are proven to converge in a Sobolev space, we must still choose a norm and we rely upon traditional metric dependent definitions.  
To begin, equip $M$ with a Riemannian metric $g:TM \oplus TM \to \mathbb{R}$.
The metric, $g$, induces a positive density $\mu_g$, known as the \emph{metric density} and an inner-product on $C^\infty(M)$
given by:
\begin{align}
	\langle f_1 , f_1 \rangle_{g} = \int \overline{f_1} \cdot f_2 \mu_g.
\end{align}
The metric also induces an elliptic operator, known as the Laplace-Beltrami operator $\Delta: C^{\infty}(M) \to C^{\infty}(M)$, which is negative-semidefinite (i.e. $\int f \, \Delta f \, \mu_{g} \leq 0$ for all $f \in C^{2}(M)$).
If $M$ is compact, then $L^2(M ; \mu_g) \cong L^2(M)$ is a separable Hilbert space and the Helmholtz operator, $1 - \Delta$, is a positive definite operator with a discrete spectrum~\cite{Taylor1974}.
For any $s \geq 0$ we may define the \emph{Sobolev norm}:
\begin{align}
	\| \psi \|_{s,2} =  \left( \langle f , (1-\Delta)^s \cdot  f \rangle_{g} \right)^{1/2}
\end{align}
where $f$ and $\psi$ are related by $\psi = f \sqrt{\mu_{g}}$.
Then  we define $H^s(M ;g)$ as the completion of $\sqrt{\Dens(M)}$ with respect to the $\| \cdot \|_{s,2}$ norm.  
Such a definition is isomorphic, in the category of topological vector-spaces, to the one provided in~\cite{Hebey1999}.
In order to prove this claim, observe that it holds for bounded sets in $\mathbb{R}^{n}$, and then apply a partition of unity argument to obtain the desired equivalence on manifolds.
In particular, note that $H^0(M;g) = L^2(M)$.  It is notable that $H^{s}(M;g)$ as a topological vector-space is actually not metric dependent ~\cite[Proposition 2.2]{Hebey1999}.
However, the norm $\| \cdot \|_{s,2}$ is metric dependent.

\begin{proposition}[Sobelev Embedding Theorem ~\cite{Taylor1974}] \label{prop:compact_embedding}
	Let $(M,g)$ be a compact Riemmanian manifold.  If $s > t \geq 0$ then $H^s(M;g)$ is compactly embedded within $H^t(M,g)$.
\end{proposition}
\begin{proof}
	Let $e_0, e_1,\dots$ be the Hilbert basis for $L^2(M;\mu_g)$ which diagonalizes $\Delta$
	in the sense that $\Delta e_i = \lambda_i e_i$ for a sequence $0 = \lambda_0 \leq \lambda_1 \leq \lambda_2 \leq \cdots$.
	The operator $(1+\Delta)^s$ is given by
	\begin{align}
		(1+\Delta)^s \cdot f =  \sum_{i} e_i (1+\lambda_i)^s \langle e_i , f \rangle_g,
	\end{align}
	and so $\{ (1+ \lambda_i)^{-s} e_i \}_{i=1}^{\infty}$ is a Hilbert basis for $H^s(M;g)$.
	
	Let us call $e_i^{(s)} = (1+ \lambda_i)^{-s} e_i$.
	The embedding of $H^s(M;g)$ into $H^t(M,g)$
	is then given in terms of the respective basis elements by $e_i^{(s)} \mapsto (1+\lambda_i)^{-(s-t)}e_i^{(t)}$.
	As $s > t$ and $\lambda_i \to +\infty$, we see that 
	this embedding is a compact operator ~\cite[Proposition 4.6]{Conway1990}.
\end{proof}


\section{Quantization} \label{sec:quantization}
In physics, ``quantization'' refers to the process of substituting certain physically relevant functions with operators on a Hilbert space, while attempting to preserve the symmetries and conservation laws of the classical theory~\cite{BatesWeinstein1997,Dirac2013,GuilleminSternberg1970}.
In this section, we quantize \eqref{eq:function pde} and \eqref{eq:density pde} by replacing functions and densities with bounded and trace-class operators on $L^{2}(M)$.
This is useful in Section \ref{sec:discretization} when we discretize.

To begin, let us quantize the space of continuous real-valued functions $C(M)$.
For each $f \in C(M)$, there is a unique bounded Hermitian operator, $H_{f} : L^{2}(M) \to L^{2}(M)$ given by scalar multiplication.
That is to say $(H_{f} \cdot \psi) (x) = f(x) \psi(x)$ for any $\psi \in L^{2}(M)$.
By inspection one can observe that the map ``$f \mapsto H_{f}$'' is injective and preserves the algebra of $C(M)$ because $H_{f\cdot g + h} = H_{f} \cdot H_{g} + H_{h}$ and $\| H_{f} \|_{op} = \| f \|_{\infty}$.

Similarly, (and in the opposite direction) for any trace class operator $A$ there is a unique distribution $\rho_{A} \in C(M)'$ such that: 
\begin{align}
	 \int f \, \rho_{A} = \Tr ( H_{f}^{\dagger} \cdot A )
\end{align}
for any $f \in C(M)$.
More generally, for any $A$ in the dual-space $B( L^{2}(M) )^{*}$, there is a $\rho_{A} \in C(M)'$ such that $\langle f , \rho_{A} \rangle = \langle H_{f} , A \rangle$.
The map ``$A \mapsto \rho_{A}$'' is merely the adjoint of the injection ``$f \mapsto H_{f}$''. Therefore ``$A \mapsto \rho_{A}$'' is surjective.

We can now convert the evolution PDEs \eqref{eq:function pde} and \eqref{eq:density pde} into ODEs of operators on $L^{2}(M)$.

\begin{theorem} \label{thm:quantize}
	Let $X(t) \in \mathfrak{X}(M)$ be a time-dependent vector-field.
	Then $f$ satisfies \eqref{eq:function pde}
	if and only if $H_{f}$ satisfies
	\begin{align}
		\frac{d H_{f} }{dt} + [ H_{f} , \pounds_{X} ] = 0. \label{eq:quantum observable ode}
	\end{align}
	If $A$ is trace-class and satisfies
	\begin{align}
		\frac{dA}{dt} + [ A , \pounds_{X} ] = 0, \label{eq:quantum density ode}
	\end{align}
	then $\rho_{A}$ satisfies \eqref{eq:density pde}.
	Finally, if $\psi$ satisfies \eqref{eq:half density pde}, then $\rho_{A} = \psi^{2}$ satisfies \eqref{eq:density pde} and $\psi \otimes \psi^{\dagger}$ satisfies \eqref{eq:quantum density ode}.
\end{theorem}

\begin{proof}
	Let $f$ satisfy \eqref{eq:function pde}.
	For an arbitrary $\psi \in L^{2}(M)$ we observe that $[ H_{f} , \pounds_{X}] \cdot \psi$ is given in coordinates by:
	\begin{align}
		\left( [ H_{f}, \pounds_{X} ] \cdot \psi \right)(x) &= f(x) \left( \frac{1}{2} X^{j} \pder{\psi}{x^{j}} + \frac{1}{2} \pder{}{x^{j}} ( \psi \, X^{j} ) \right) (x) \\
			&\quad- \frac{1}{2} X^{j} \left. \pder{}{x^{j}} \right|_{x}( f \psi)  + \frac{1}{2} \left. \pder{}{x^{j}} \right|_{x}(f \psi \, X^{j} )
	\end{align}
	where we have used \eqref{eq:representation}.  
	Application of the product rule to each of these terms yields a number of cancellations and we find:
	\begin{align}
		[ H_{f} , \pounds_{X} ] \cdot \psi = - X^{j} \pder{f}{x^{j}} \psi = (\partial_{t} f )\psi = \frac{d H_{f} }{dt} \cdot \psi.
	\end{align}
	As $\psi$ is arbitrary, we have shown that $H_{f}$ satisfies \eqref{eq:quantum observable ode}.
	Each line of reasoning is reversible, and so we have proven the converse as well.
	
	In order to handle densities note that $\langle f , \rho \rangle$ is constant in time when $f$ and $\rho$ satisfy \eqref{eq:function pde} and \eqref{eq:density pde}, respectively.
	By the definition of $\rho_{A}$, $\Tr( H_{f}^{\dagger} \cdot A) = \langle f , \rho_{A} \rangle$.
	Therefore:
	\begin{align}
		0 = \frac{d}{dt} \left( \Tr( H_{f}^{\dagger} \cdot A ) \right) = \Tr \left( \frac{d H_{f}^{\dagger}}{dt} \cdot A + H_{f}^{\dagger} \cdot \frac{d A}{dt} \right).
	\end{align}
	As was just shown, $\frac{dH_{f}}{dt} =  - [H_{f} , \pounds_{X} ]$ so:
	\begin{align}
		0 = \Tr \left( - [H_{f} , \pounds_{X} ]^{\dagger} \cdot A + H_{f}^{\dagger} \cdot \frac{d A}{dt} \right) 
		= \Tr( - \pounds_{X}^{\dagger} H_{f}^{\dagger} \cdot A + H_{f}^{\dagger} \pounds_{X}^{\dagger} \cdot A + H_{f}^{\dagger}  \cdot \frac{dA}{dt} ).
	\end{align}
	Upon noting that $\pounds_{X}^{\dagger} = - \pounds_{X}$ and that $\Tr( a b c) = \Tr( bc a)$:
	\begin{align}
		0 = \Tr \left( H_{f}^{\dagger}( [\hat{\rho} , \pounds_{X} ] + \frac{d \hat{\rho}}{dt} ) \right).
	\end{align}
	As $H_f$ was chosen arbitrarily, the desired result follows.
	Again, this line of reasoning is reversible.

	Lastly, if $\psi$ satisfies \eqref{eq:half density pde} and $\rho = \psi^{2}$ then we see
	\begin{align}
		\partial_{t} \rho &= \partial_{t} ( \psi^{2}) = 2 (\partial_{t} \psi ) \, \psi \\
		&= 2 \left( - \frac{1}{2} \partial_{i} (\psi X^{i}) - \frac{1}{2} X^{i} \partial_{i} \psi \right) \psi 
		= \left( - X^{i} \partial_{i} \psi - \psi \partial_{i} X^{i}  \right) \psi \\
		&= - 2 \psi (\partial_{i} \psi) X^{i} - (\partial_{i}X^{i}) \psi^{2}
		= - \partial_{i}(\rho) X^{i} - (\partial_{i}X^{i}) \rho = - \partial_{i} ( \rho X^{i}).
	\end{align}
\end{proof}

The benefit of using \eqref{eq:quantum observable ode} and \eqref{eq:quantum density ode} to represent the PDEs of concern is that \eqref{eq:quantum observable ode} and \eqref{eq:quantum density ode} may be discretized using a standard least squares projections on $L^{2}(M)$ without sacrificing qualitative accuracy.

\section{Discretization} \label{sec:discretization}
This section presents the numerical algorithms for solving \eqref{eq:function pde} and \eqref{eq:density pde}.
The basic ingredient for all the algorithms in this section are a Hilbert basis and an ODE solver.
Denote a Hilbert basis by $\{ e_{0}, e_{1},\dots \}$ for $L^{2}(M)$.
For example, for a Riemannian metric, $g$, if $\{ f_{0}, f_{1},\dots \}$ denote eigen-functions of the Laplace operator, then $\{ E_{k} = f_{k} \sqrt{\mu_{g}} \mid k \in \mathbb{N} \}$ forms a smooth Hilbert basis for $L^{2}(M)$ where $\mu_{g}$ denotes the Riemannian density.  We call $\{ E_{k} \}$ the Fourier basis.
To ensure convergence, we assume:
\begin{ass} \label{ass:basis}
	Our basis $\{ e_{k} \}$ is such that there exists a metric $g$ for which the unitary transformation which sends the basis $\{ e_{k} \}$ to the Fourier basis is bounded with respect to the $\| \cdot \|_{s,2}$-norm for some $s > 1$.
\end{ass}


In this section we provide a semi-discretization of \eqref{eq:function pde} and \eqref{eq:density pde}.
Just as a note to the reader, a ``semi-discretization''  of the PDE $\partial_{t} \phi + F(\phi) = 0$ for some partial differential operator, $F$, is just a discretization of $F$ which 
converts the PDE into an ODE~\cite{Gottlieb1977numerical}.
In particular, we assume access to solvers of finite dimensional ODEs, denoted ``$\OdeSolve$.''
In practice any ODE solver such as Euler's method, Runge-Kutta, or even well tested software such as ~\cite{VODE} could be used to compute such solutions.
Most notably, the method of~\cite{Calvo1997} is specialized to isospectral flows such as  \eqref{eq:quantum observable ode} and \eqref{eq:quantum density ode} by using discrete-time isospectral flows.
More explicitly, let $\OdeSolve ( F , x_{0} , t )$ denote the numerically computed solution $x(t)$ to the ODE ``$\dot{x} = F(x)$'' at time $t \in \mathbb{R}$, with initial condition $x_{0} \in M$.
Before constructing an algorithm to spectrally discretize \eqref{eq:function pde} and \eqref{eq:density pde} in a qualitatively accurate manner, we first solve \eqref{eq:half density pde} using a standard spectral discretization in Algorithm \ref{alg:half density}~\cite{Boyd2001,NumericalRecipes}.

\begin{algorithm}
	\KwIn{$\psi(0) \in L^{2}(M)$, $t \in \mathbb{R}$, $N \in \mathbb{N}$.}
	initialize $z(0) \in \mathbb{C}^{N}$.\;
	initialize $X_{N} \in \mathbb{C}^{N \times N}$\;
	\For{ $i = 1 , \dots , N$}{
		$[z(0)]_{i} = \int_{M} \bar{e}_{i}(x) \psi(0,x)$\;
		\For{$j=1,\dots,N$}{
			$[X_{N}]\indices{^{i}_{j}} = \frac{1}{2} \int_{M} \bar{e}_{i}(x) ( X^{\alpha} \partial_{\alpha} e_{j} + \partial_{\alpha}( X^{\alpha} e_{j}) )(x)$
		} 
	}
	initialize the function $F: \mathbb{C}^{N} \to \mathbb{C}^{N}$ given by $F(z) = X_{N} \cdot z$.\;
	$z(t) = \operatorname{OdeSolve}(  F , z(0) , t)$\;
	\KwOut{ $\psi_{N}(t) = \sum_{i=1}^{n} [z(t)]_{i} e_{i}$.}
	\caption{A spectral discretization to solve \eqref{eq:half density pde} for half densities.} \label{alg:half density}
\end{algorithm}

To summarize, Algorithm \ref{alg:half density} produces a half-density $\psi_{N}(t_{k}) \in V_{N}$ by projecting \eqref{eq:half density pde} to $V_N$.
This projection is done by constructing the operator $X_{N} = \pi_{N} \circ \pounds_{X} |_{V_{N}}: V_{N} \to V_{N}$.
In Section \ref{sec:analysis} we prove that $\psi_{N}(t_{k})$ converges to the solution of \eqref{eq:half density pde} as $N \to \infty$.
We see that $\psi_{N}(t)$ evolves by unitary transformations, just as the exact solution to \eqref{eq:half density pde} does.
This correspondence is key in providing the qualitative accuracy of algorithms that follow, so we formally state it here.
\begin{proposition} \label{prop:unitary}
	The output of Algorithm \ref{alg:half density} is given by $U_{N}(t_{k}) \cdot \psi_{N}(0)$ when $\psi(0) \in L^2(M)$ is the input to Algorithm \ref{alg:half density} where $\psi_{N}(0) = \pi_{N}( \psi(0) )$ and $U_{N}(t)$ is the unitary operator as in Proposition \ref{prop:stone} generated by $X_{N}$.
\end{proposition}
\begin{proof}
	The operator $X_{N}$ in Algorithm \ref{alg:half density} is anti-Hermitian on $V_{N}$.
	It therefore generates a unitary action on $V_{N} \subset L^{2}(M)$ when inserted into $\operatorname{OdeSolve}$.
\end{proof}

Before continuing, we briefly state a sparsity result that aides in selecting a basis.
We say an operator $A : L^{2}(M) \to L^{2}(M)$ is \emph{sparse banded diagonal} with respect to a Hilbert basis $\{ e_{0} , e_{1},\dots\}$ if there exists an integer $W \in \mathbb{N}$
such that $A(e_{i})$ is a finite sum elements of the form $e_{i + \delta_{j}}$ for fewer than $W$ offsets $\delta_{j}$ for $i = 0,1,2, \dots$.
\begin{theorem} \label{thm:sparsity}
	Let $x^{1},\dots,x^{n}$ be a dense coordinate chart for $M$ on some dense open set\footnote{Such a chart always exists on a compact manifold by choosing a Riemannian metric and extending a Riemannian exponential chart to the cut-locus~\cite{Sakai1996,MO_dense_charts}. },
	then $e_{k}= f_{k} \, \sqrt{\mu}$ for functions $f_{k} \in L^2(M; \mu)$ where $\mu = | dx^{1} \wedge \cdots \wedge dx^{n}|$ (see Proposition \ref{prop:non canonical}).
	If $\rho( \pder{}{x^{j}})$ and $H_{f_{k}}$ are sparse banded diagonal, 
	and if the vector-field $X$ is given in local coordinates by $X^{i} = \sum_{k} c_{k}^{i} f_{k}$ with fewer than $W>0$ of $c^{i}_{k}$'s being non-zero for each $i=1,\dots,n$, then the matrix $X_{N}$ in Algorithm \ref{alg:half density} is sparse banded diagonal and the sparsity of $X_{N}$ is $\mathcal{O}( W / N )$.
\end{theorem}
\begin{proof}
The result follows directly from counting.
\end{proof}

Theorem \ref{thm:sparsity} suggests selecting a basis where $W$ is small, or at least finite.
For example, if $M$ were a torus, and the vector-field was made up of a finite number of sinusoids, then a Fourier basis would yield a $W$ equal to the maximum number of terms along all dimensions.

By Theorem \ref{thm:quantize}, the square of the result of Algorithm \ref{alg:half density} is a numerical solution to \eqref{eq:density pde}.
We can use this to produce a numerical scheme to \eqref{eq:density pde} by finding the square root of a density.
Given a $\rho \in \Dens(M)$, let $\rho^{+}$ denote the positive part and $\rho^{-}$ denote the negative part so that $\rho = \rho^{+} - \rho^{-}$, 
then $\psi = \sqrt{\rho^{+}} - i \sqrt{\rho^{-}}$ is a square root of $\rho$ since $\rho = \psi^{2}$.
This yields Algorithm \ref{alg:density} to spectrally discretize \eqref{eq:density pde} in a qualitatively accurate manner for densities which admit a square root.

\begin{algorithm}[H] 
	\KwData{$\rho(0) \in L^{1}(M), t \in \mathbb{R},N \in \mathbb{N}$.}
	Initialize $\psi(0) =  \sqrt{\rho^{+}(0)} - i \sqrt{\rho^{-}(0)}$\;
	Set $\psi_{N}(t) =  \operatorname{ Algorithm\_ \ref{alg:half density} }( \psi(0) , t, N )$\;
	\KwOut{ $\rho_{N}(t , x) = \psi_{N}(t , x)^{2}$.}
	\caption{A spectral discretization to solve \eqref{eq:density pde} for densities} \label{alg:density}
\end{algorithm}

Alternatively, we could have considered the trace-class operator $A_{N}(t_{k}) = \psi_{N}(t_{k}) \otimes \psi_{N}(t_{k})^{\dagger}$ as an output. 
This would be an numerical solution to \eqref{eq:quantum density ode}, and would be related to our original output in that $\rho_{N}(t_{k}) = \rho_{A_{N}(t_{k})}$.
Finally, we present an algorithm to solve \eqref{eq:quantum observable ode} (in lieu of solving \eqref{eq:function pde}).
This algorithm is presented for theoretical interest at the moment.

\begin{algorithm}[H]
	\KwData{$f(0) \in C(M), t \in \mathbb{R}, N \in \mathbb{N}$.}	
	initialize $F_{N}(0), X_{n} \in \mathbb{C}^{N \times N}$.\;
	initialize the linear map $B: \mathbb{C}^{N \times N} \to \mathbb{C}^{N \times N}$
	given by $B(H) = - [H , X_{N}]$.\;
	\For{ $i , j = 1 , \dots , N$ }{
		$[F_{N}(0)]\indices{^i_{j}} =  \int_{M} \bar{e}_{i}(x) f(x) e_{j}(x)$ \;
		$[X_{N}]\indices{^i_{j}} = \frac{1}{2} \int_{M} \bar{e}_{i}(x) ( X^{\alpha} \partial_{\alpha} e_{j} + \partial_{\alpha}( X^{\alpha} e_{j}) )(x)$ \;
	}
	$F_{N}(t) = \OdeSolve( B , F_{N}(0) , t)$\;	
	\KwOut{ The (compact) operator $H_{f,N}(t_{k}) = \sum_{i,j=1}^{N}[ F_{N}(t_{k}) ] \indices{^{i}_{j}} e^{i} \otimes e_{j}^{\dagger}$. }
	\caption{A spectral discretization to solve \eqref{eq:quantum observable ode} for functions} \label{alg:function}
\end{algorithm}

We find that the ouput of Algorithm \ref{alg:function} bears algebraic similarities similarities to the exact solution to the infinite dimensional ODE, \eqref{eq:quantum observable ode} (which is isomorphic to \eqref{eq:function pde} by Theorem \ref{thm:quantize}).
This is stated in a proposition analogous to Proposition \ref{prop:unitary}.

\begin{proposition} \label{prop:isospectral}
$H_{f,N}(t) = U_{N}(t) \cdot H_{f,N}(0) \cdot U_{N}(t)^{\dagger}$ for any $t \in \mathbb{R}$.
Moreover, $U_{N}(t)$ is identical to the unitary transformation of Proposition \ref{prop:unitary}.
Lastly, the exact solution of \eqref{eq:quantum observable ode} is of the form $H_{f}(t) = U(t) \cdot H_{f}(0)  \cdot U(t)^{\dagger}$ as well.
\end{proposition}
\begin{proof}
	This follows from the fact that algorithm outputs the solution to an isospectral flow ``$\dot{F}_{N} + [F_{N} , X_{N}]$''
	where $X_{N}$ is anti-Hermitian and that the $H_{f}$ satisfies the isospectral flow \eqref{eq:quantum observable ode}.
\end{proof}

\section{Error analysis} \label{sec:analysis}

In this sections we derive convergence rates.
We find that the error bound for Algorithm \ref{alg:half density} induces error bounds for the Algorithms \ref{alg:density} and \ref{alg:function}.
Therefore, we first derive a useful error bound for Algorithm \ref{alg:half density}.
Our proof is a generalization of the convergence proof in ~\cite{Pasciak1980}, where \eqref{eq:half density pde} is studied (modulo a factor of two time rescaling) on the torus.
We begin by proving an approximation bound.
In all that follows, let $\pi_{N}: L^{2}(M) \to V_{N}$ denote the orthogonal projection.

\begin{proposition} \label{prop:approximation}
	If $\psi \in H^{\bar{s}}(M)$ and $\bar{s} > s \geq 0$,
	then
	\begin{align}
		\| \psi - \pi_{N}(\psi) \|_{s,2} <  \frac{d \, C_{\bar{s},s} }{ \bar{s}-s} \| \psi \|_{\bar{s} , 2} \, N^{-2(\bar{s}-s)/n}
	\end{align}
	for some constant $C_{\bar{s},s}$ and $d = \dim(M)$.
\end{proposition}
\begin{proof}
	We can assume that $e_{1},e_{2},\dots$ is a Fourier basis.
	The results are unchanged upon applying Assumption \ref{ass:basis} and converting to the Fourier basis.
	Any $\psi \in H^{s}(M;g)$ can expanded as $\psi = \hat{\psi}_{k} e_{k}$ where $\hat{\psi}_{k} = \langle e_{k} \mid \psi \rangle$.
	As $\psi \in H^{s}(M;g)$ it follows that
	\begin{align}
		\| \psi \|^{2}_{\bar{s},2} = \sum_{k=0}^{\infty} \left| \hat{\psi}_{k} \right|^{2} (1+\lambda_{k})^{\bar{s}} < \infty.
		\label{eq:propapprox_eq1}
	\end{align}
	A corollary of Weyl's asymptotic formula is that $\lambda_{k}$ is $\mathcal{O}( k^{2/n})$ for large $k$ ~\cite[page 155]{Chavel1984}.
	After substitution of this asymptotic result into \eqref{eq:propapprox_eq1} for large $k$, we see that $|\hat{\psi}_{k}|^{2}$ is asymptotically dominated by  $C k^{-1- 2\bar{s}/n}$ for some constant $C$.
	For sufficiently large $N$ we find
	\begin{align}
		\| \psi - \pi_{N}(\psi) \|_{s,2} = \sum_{k>N} (1+\lambda_{k})^{s} |\hat{\psi}_{k}|^{2} \leq C \sum_{k>N} (1+\lambda_{k})^{s} k^{-1- 2\bar{s}/n}
	\end{align}
	and by another application of the Weyl formula
	\begin{align}
		\| \psi - \pi_{N}(\psi) \|_{s,2} \leq \tilde{C} \sum_{k>N} \frac{1}{k^{1+2(\bar{s}-s)/n}} \leq C_{s,\bar{s}}  \frac{d }{ \bar{s}-s} N^{-2(\bar{s}-s)/n}.
	\end{align}
	Where the last inequality is derived by bounding the infinite sum with an integral.
\end{proof}


With this error bound for the approximation error we can derive an error bound for Algorithm \ref{alg:half density}:

\begin{theorem} \label{thm:half density convergence}
	Let $\psi(0) \in H^{\bar{s}}(M)$ for $\bar{s} > s > 1$.
	Let $T > 0$ and $t \in [0,T]$.
	Let $\psi(t)$ be denote the solution to \eqref{eq:half density pde}
	with initial condition $\psi(0)$.
	Finally, let $\psi_{N}(t)$ be the output of Algorithm \ref{alg:half density}
	with respect to the inputs $\psi(0), t , N$ for some $N \in \mathbb{N}$.
	Then the error $\varepsilon_{N}(t) := \| \psi(t) - \psi_{N}(t) \|_{s,2}$ satisfies:
	\begin{align}
		\varepsilon_{N}(t) \leq \| \psi(0) \|_{\bar{s},2} \, K_{T} \left( N^{-2(s-1)} t+  \frac{n}{\bar{s}-s} N^{-2(\bar{s}-s)/n} \right) e^{C_{T} t}
	\end{align}
	where $K_{T}$ and $C_{T}$ are positive and constant with respect to $N$,$s$, and $\bar{s}$.
	In particular for $s = (\bar{s}+1)/2$:
	\begin{align}
			\varepsilon_{N}(t) \leq \| \psi(0) \|_{\bar{s},2} \, K_{T} \left( N^{1-\bar{s}} t+  \frac{n}{\bar{s}-1} N^{(1-\bar{s})/n} \right) e^{C_{T} t}.
	\end{align}
\end{theorem}

\noindent To prove Theorem \ref{thm:half density convergence}, we need a perturbed version of Gronwall's inequality:
\begin{lemma} \label{lem:Gronwall}
If $\frac{du}{dt} \leq Ku + \epsilon$ for some $K>0$ then $u(t) \leq (\epsilon t + u(0) ) e^{Kt}$.
\end{lemma}
\begin{proof}
	Let $w (t)= u (t) e^{-Kt}$.  Then for $t \geq 0$ we find
	\begin{align}
		\frac{dw}{dt} = \frac{du}{dt} e^{-Kt} - K w \leq (Ku+\epsilon) e^{-Kt} - Kw = \epsilon e^{-Kt} \leq \epsilon
	\end{align}
	Thus $w(t) \leq \epsilon t + w(0) = \epsilon t + u(0)$.
\end{proof}

Now we can prove Theorem \ref{thm:half density convergence}:

\begin{proof}[Proof of Theorem \ref{thm:half density convergence}]
	Note that $\frac{d\varepsilon_{N}}{dt} = \frac{1}{2\varepsilon_{N}} \langle  \psi - \tilde{\psi} \mid (1+\Delta)^{s} \frac{d}{dt} ( \psi -\tilde{\psi} )\rangle$
	By the Cauchy-Schwarz inequality
	\begin{align}
		\frac{d\varepsilon_{N}}{dt} &\leq  \frac{1}{2} \| \pounds_{X}[\psi] - \pi_{N}(\pounds_{X}[\psi_{N}]) \|_{s,2} \\
		&= \| \pounds_{X}[\psi] - \pi_{N}( \pounds_{X}[\psi- (\psi-\psi_{N})]) \|_{s,2}.
	\intertext{By the triangle inequality and the definition of the operator norm:}
		\frac{d\varepsilon_{N}}{dt}&\leq \| (1-\pi_{N}) \|_{H^{s-1},op} \, \|X \|_{H^{s},op} \, \| \psi \|_{s,2} + \| \pi_{N} \|_{op} \, \|X \|_{H^{s},op} \, \varepsilon_{N}
	\end{align}
	By Proposition \ref{prop:stone} we observe that $\psi(t)$ is related to $\psi_{0}$ through the flow of $X$ which is a $C^{k}$-diffeomorphism if $X$ is $C^{k}$.
	From the local expression Proposition \ref{prop:stone} in we can observe that $\| \psi(t) \|_{s,2}$ is bounded by a scalar multiple of $\| \psi_{0} \|_{s,2}$.
	Thus we may write the above bound in the form
	\begin{align}
		\frac{d\varepsilon_{N}}{dt} \leq K' \, \| 1- \pi_{N} \|_{H^{s-1},op} \, \| \psi_{0}\|_{s,2}+ C_{T} \varepsilon_{N}
	\end{align}
	for constants $C_{T}$ and $K'$.
	As $s > 1$, for sufficiently large $N$ we can compute that $\| 1-\pi_{N} \|_{H^{s-1},op} \leq (1+\lambda_{N+1})^{-(s-1)}$ where $\lambda_{N}$ denotes the $N$th eigenvalue of the Laplace operator.
	This is accomplished by observing the operator $1-\pi_{N}$ in a Fourier basis and applying to appropriate norms.
	By Weyl's asymptotic formula ~\cite[Theorem B.2]{Chavel1984}, $\lambda_{N}$ asymptotically behaves like $N^{2/n}$.
	Therefore by Lemma \ref{lem:Gronwall} with $\epsilon = C_{T} n^{-2(s-1) / d} \, \| \psi_{0}\|_{s,2}$:
	\begin{align}
		\varepsilon_{N}(t) \leq ( K' N^{-2(s-1) / n} \| \psi_{0} \|_{s,2} t+  \varepsilon_{N}(0) ) e^{C_{T} t}.
	\end{align}
	That $\varepsilon_{N}(0)$ behaves as $K'' \| \psi_{0} \|_{\bar{s},2} n^{-2(\bar{s}-s)/d}$ is a re-statement of Proposition \ref{prop:approximation}.
	We then set $K_{T} = \max(K', K'')$.
\end{proof}

Having derived an error bound for Algorithm \ref{alg:half density}, we can derive an error bound for Algorithm \ref{alg:density}.

\begin{theorem} \label{thm:density convergence}
	Let $\rho(0)$ be a distribution in $W^{\bar{s},1}(M)$ for $\bar{s} > s >1$.
	Let $T > 0$ and $t \in [0,T]$ be fixed.
	Let $\rho(t)$ be the solution of \eqref{eq:density pde} at time $t$.
	Finally, let $\rho_{N}(t)$ be the output of Algorithm \ref{alg:density} with respect to the input $(\rho(0), t , N)$ for some $N \in \mathbb{N}$.
	Then:
	\begin{align}
		\| \rho(t) - \rho_{N}(t) \|_{1} \leq \| \rho(0) \|_{\bar{s},1} \, K \left( N^{-2(s-1)} t+  \frac{d}{\bar{s}-s} N^{-2(\bar{s}-s)/n} \right) e^{C_{T} t}
	\end{align}
	where $K$ is constant with respect to $N$, and $C_{T}$ is the same constant as in Theorem \ref{thm:half density convergence}.
\end{theorem}

\begin{proof}
	Without loss of generality, assume that $\rho$ is non-negative (otherwise split it into its non-negative and non-positive components).
	Let $\psi \in L^{2}(M)$ be such that $\rho =  \psi ^{2}$, as described in Algorithm \ref{alg:density}.
	It follows that $\psi \in H^{s}(M)$ and we compute
	\begin{align}
		\| \rho(t) - \rho_{N}(t) \|_{1} = \int_{M} | \rho(t) - \rho_{N}(t)| = \int_{M} | \psi^{2} - \psi_{N}^{2} |
	\end{align}
	If we let $\phi_{N} = \psi - \psi_{N}$ then we can re-write the above as
	\begin{align}
		\| \rho(t) - \rho_{n}(t) \|_{1}  &= \int_{M} | \psi^{2} - (\psi - \phi_{N})^{2} | = \int_{M} | 2 \psi \phi_{N} - \phi_{N}^{2} | \\
			&\leq 2 \| \psi \|_{2} \| \phi_{N}\|_{2} + \| \phi_{N} \|_{2}^{2} = 2 \| \rho \|_{1}^{1/2} \cdot \| \phi_{N} \|_{2} + \| \phi_{N} \|_{2}^{2}
	\end{align}
	Above we have applied Holder's inequality to $L^{2}(M)$, which still holds upon using the isometry in Proposition \ref{prop:non canonical}.
	Theorem \ref{thm:half density convergence} provides a bound for $\| \phi_{N} \|$.
	Substitution of this bound into the above inequality yields the theorem.
\end{proof}

Finally, we prove that Algorithm \ref{alg:function} converges to a solution of \eqref{eq:quantum observable ode}, which is equivalent to a solution of \eqref{eq:function pde} courtesy of Theorem \ref{thm:quantize}:

\begin{proposition} \label{prop:function approximation}
	Let $f \in C^{k}(M)$ and let $H_{f,N} = \pi_{N} \circ H_{f} \circ \pi_{N}$.  Then
	\begin{align}
		\| H_{f} - H_{f,N} \|_{H^{s},op} \leq D \frac{n}{s} N^{-2s/n} \| \hat{f} \|_{op}
	\end{align}
	where $s > k \geq 1$, and $D$ is constant.
\end{proposition}
\begin{proof}
	Let $\pi_{N}^{\perp} = 1 - \pi_{N}$.  By Proposition \ref{prop:approximation} we know that
	\begin{align}
		\| \pi_{N}^{\perp}(\psi) \|_{2} \leq \frac{n}{s} N^{-2s/n} \| \psi \|_{s,2} \label{eq:hot inequality}
	\end{align}
	for $s>0$, then:
	\begin{align*}
		\| H_{f} - H_{f,N} \|_{H^{s},op} &= \sup_{\| \psi \|_{s,2}=1} \langle \psi \mid H_{f} - H_{f,N} \mid \psi \rangle \\
			&= \sup_{\| \psi \|_{s,2}=1} \left( \langle \psi \mid H_{f}  \mid \psi \rangle - \langle \psi - \pi_{N}^{\perp}(\psi) \mid H_{f} \mid \psi - \pi_{N}^{\perp}(\psi) \rangle \right) \\
			&= \sup_{\| \psi \|_{s,2}=1} \left( 2 \Re \langle \pi_{N}^{\perp}(\psi) \mid H_{f} \mid \psi \rangle - \langle \pi_{N}^{\perp}(\psi) \mid H_{f} \mid \pi_{N}^{\perp}(\psi) \rangle \right) \\
			&\leq \sup_{\| \psi \|_{s,2}=1}  ( \| \pi_{N}^{\perp}(\psi) \|_{2}- \| \pi_{N}^{\perp}(\psi) \|_{2}^{2} ) \| H_{f} \|_{op} 
	\end{align*}
	By \eqref{eq:hot inequality} the result follows.
\end{proof}

\begin{theorem} \label{thm:function convergence}
	Let $T > 0$ and $t \in [0,T]$ be fixed.
	Let $f(t)$ denote the solution to \eqref{eq:function pde} at time $t$ with initial condition $f(0) \in C^{k}(M)$.
	Let $H_{f,N}(t)$ denote the output of Algorithm \ref{alg:function} with respect to the inputs $(f(0) , t  , N)$ for some $N \in \mathbb{N}$.
	Then:
	\begin{align}
		\| H_{f(t)} - H_{f,N}(t) \|_{H^{s},op} &\leq D \frac{n}{s} N^{-2s/n} \| H_{f(t)} \|_{op} \\
			&\quad +  K_{T} \| H_{f,N}(t) \|_{op} \left( N^{1-s} + \frac{2n}{s -1} N^{(1-s)/n} \right) e^{C_{T}t} \nonumber
	\end{align}
	for the same constant $D$ as in Proposition \ref{prop:function approximation} and the same constants $C_{T},K_{T}$ as in Theorem \ref{thm:half density convergence}.
\end{theorem}

\begin{proof}
	We find
	\begin{align*}
		\| H_{f(t)} - H_{f,N}(t) \|_{H^{s},op} &= \sup_{\| \psi \|_{s,2} =1} \langle \psi \mid H_{f(t)} - H_{f,N}(t) \mid \psi \rangle
	\end{align*}
	In light of Proposition \ref{prop:isospectral} we find
	\begin{align*}
			\quad = \sup_{ \| \psi \|_{s,2} =1} \langle \psi \mid U(t) \cdot H_{f(0)} \cdot U(t)^{\dagger} - U_{N}(t) \cdot H_{f,N}(0) \cdot U_{N}(t)^{\dagger} \mid \psi \rangle.
	\end{align*}
	The output of Algorithm \ref{alg:function} indicates that $H_{f,N}(0) = \pi_{N}^{\perp} \circ H_{f(0)} \circ \pi_{N}^{\perp}$.
	Therefore, the above inline equation becomes
	\begin{align*}
			= \sup_{ \| \psi \|_{s,2} =1 } \langle U(t)^{\dagger} \psi \mid H_{f(0)} \mid U(t)^{\dagger} \psi \rangle
				- \langle U_{N}(t)^{\dagger} \psi \mid \pi_{N}^{\perp} \circ H_{f(0)} \circ \pi_{N}^{\perp} \mid U_{N}(t)^{\dagger} \psi \rangle.
	\end{align*}
	and finally
	\begin{align}
			&= \sup_{ \| \psi \|_{s,2} = 1} \langle U(t)^{\dagger} \psi \mid H_{f(0)} - H_{f,N}(0) \mid U(t)^{\dagger} \psi \rangle 
			- \langle  \phi(t) \mid H_{f,N}(0) \mid  \phi(t) \rangle \label{eq:final line}
	\end{align}
	where $\phi(t) = U_{N}(t)^{\dagger} \psi - U(t)^{\dagger} \psi$.
	
	The first term is bounded by Proposition \ref{prop:function approximation}.
	To bound the second term we must bound $\phi$.
	As $U_{N}(t)^{\dagger} \psi$ is the backwards time numerical solution to \eqref{eq:half density pde} and $U(t)^{\dagger}\psi$ is the exact backward time solution to \eqref{eq:half density pde},
	Theorem \ref{thm:half density convergence} prescribes the existence of constants $K$ and $C$ such that:
	\begin{align*}
		\| \phi \|_{\underline{s},2} = \| U_{N}(t)^{\dagger} \psi - U(t)^{\dagger} \psi \|_{\underline{s},2}  \leq K \| \psi \|_{s,2} \left(  N^{-2(\underline{s}-1)} + \frac{n}{s - \underline{s}} N^{-2(s-\underline{s})/n} \right) e^{Ct} 
	\end{align*}
	for any $\underline{s} <s$.
	This expression can be simplified by noting that $\| \psi \|_{s,2} = 1$, setting $\underline{s} = (1+s)/2$, and noting that the $H^{\underline{s}}$ norm is stronger than the $L^{2}$-norm to get:
	\begin{align*}
		 \| \phi \|_{2}  \leq  K \left(  N^{1-s} + \frac{2n}{s -1} N^{(1-s)/n} \right) e^{Ct}.
	\end{align*}
	By applying the Cauchy-Schwarz inequality to \eqref{eq:final line} and our derived bound on $\phi$:
	\begin{align*}
		\| H_{f(t)} - H_{f,N}(t) \|_{H^{s},op} &\leq \| H_{f(0)} - H_{f,N}(0) \|_{H^{s},op} \\
			&\quad+  K \| H_{f,N} \|_{op} \left( N^{1-s} + \frac{2n}{s -1} N^{(1-s)/n} \right) e^{Ct}
	\end{align*}
	Upon invoking Proposition \ref{prop:function approximation} we get the desired result.
\end{proof}

\section{Qualitative Accuracy} \label{sec:qualitative}
In this section, we prove that our numerical schemes are qualitatively accurate.
We begin by illustrating the preservation of appropriate norms.
Throughout this section let $\psi_{N}(t)$, $\rho_{N}(t)$, and $H_{f,N}(t)$ denote the sequence of outputs of Algorithms \ref{alg:half density}, \ref{alg:density}, and \ref{alg:function} with respect to initial conditions $\psi(0) \in H^{s}(M;g), \rho(0) \in W^{s,1}$ and $f(0) \in C^{s}(M)$ for $N = 1,2,\dots$.
	
\begin{theorem} \label{thm:norms}
	Let $\psi,\rho,f$ denote solutions to \eqref{eq:half density pde}, \eqref{eq:density pde}, and \eqref{eq:function pde} respectively.
	Let $\psi_{N}(t),\rho_{N}(t)$, and $H_{f,N}(t)$, denote outputs from algorithms \ref{alg:half density}, \ref{alg:density}, and \ref{alg:function} respectively for a time $t < \infty$.
	Then $\| \psi_{N} (t)\|_{2}, \|\rho_{N}(t)\|_{1},$ and $\| H_{f,N}(t) \|_{op}$ are constant with respect to $t$ for arbitrary $N \in \mathbb{N}$.
	Moreover,
	\begin{align*}
		\lim_{N \to \infty} \| \psi_{N} (t) \|_{2} = \| \psi(\cdot; t) \|_{2}, \\
		\lim_{N \to \infty} \| \rho_{N}(t) \|_{nuc} = \| \rho(\cdot; t) \|_{1}, \\
		\lim_{N \to \infty} \| H_{f,N}(t) \|_{op} = \| f( \cdot ;t) \|_{\sup}.
	\end{align*}
\end{theorem}
\begin{proof}
	To prove $\| H_{f,N} \|_{op}$ is conserved note that the evolution is isospectral ~\cite{Calvo1997}.

	We have already shown that $H_{f,N}(t)$ converges to $H_{f}(t)$ in the operator norm.
	Convergence of the norms follows from the fact that $\| H_{f}(t) \|_{op} = \| f \|_{\sup}$.
	An identical approach is able to prove the desired properties for $\rho_{N}(t)$ and $\psi_{N}(t)$ as well.
\end{proof}

Theorem \ref{thm:norms} is valuable because each of the norms is naturally associated to the entity which it bounds, and these quantities are conserved for the PDEs that this paper approximates.
For example, $\| H_{f} \|_{op} = \| f \|_{\sup}$ for a function $f$, and this is constant in time when $f$ is a solution to \eqref{eq:function pde}.
A discretization constructed according to Algorithm \ref{alg:function} according to Theorem \ref{thm:norms} is constant for any $N$, no matter how small.

The full Banach algebra $C(M)$ is conserved by advection too.
This property is encoded in our discretization as well.
\begin{theorem} \label{thm:algebra}
	Let $f(x;t),g(x;t)$, and $h(x;t)$ be solutions of \eqref{eq:function pde} and let $k= f \cdot g + h $.
	Let $H_{f,N}, H_{g,N}$ and $H_{h,N}$ be numerical solutions constructed by Algorithm \ref{alg:function}, then $H_{k,N}(t) = H_{f,N}\cdot H_{g,N}+H_{h,N}$
	satisfies
	\begin{align}
		\frac{d}{dt} H_{k,N} = [ X_{N} , H_{k,N}].
	\end{align}
	Moreover, $H_{k,N}(t)$ strongly converges to $H_{k}$ as $N \to \infty$ in the operator norm on $H^{s}(M)$ when $f,g,h \in C^{s}(M)$ for $s>1$.
\end{theorem}
\begin{proof}
	By construction, the output of Algorithm \ref{alg:function} is the result of an isospectral flow, and is therefore of the form 
	\begin{align}
		H_{f,N}(t) = U_{N}(t) H_{f,N}(0) U_{N}(t)^{\dagger} \\
		H_{g,N}(t) = U_{N}(t) H_{g,N}(0) U_{N}(t)^{\dagger} \\
		H_{h,N}(t) = U_{N}(t) H_{h,N}(0) U_{N}(t)^{\dagger}.
	\end{align}
	We then observe
	\begin{align}
		H_{k,N}(t) &= U_{N}(t) H_{k,N}(0) U_{N}(t)^{\dagger} = U(t)\left( H_{f,N}(0) H_{g,N}(0) + H_{h,N}(0) \right) U(t)^{\dagger} \\
			&= U(t)H_{f,N}(0) U(t)^{\dagger} U(t) H_{g,N}(0)U(t)^{\dagger} + U(t) H_{h,N}(0) U(t)^{\dagger} \\
			&=H_{f,N}(t)H_{g,N}(t) + H_{h,N}(t).
	\end{align}
	Differentiation in time implies the desired result.
	Convergence follows from Theorem \ref{thm:function convergence}.
\end{proof}

Finally, the duality between functions and densities is preserved by advection.  If $f$ satisfies \eqref{eq:function pde} and $\rho $ satisfies \eqref{eq:density pde} then $\int f \rho$ is conserved in time.
Algorithms \ref{alg:density} and \ref{alg:function} satisfy this same equality:
\begin{theorem}
	For each $N \in \mathbb{N}$, $\Tr( H_{f,N} A_{\rho,N} (t))$ is constant in time where $A_{\rho,N} (t)= \psi_{N}(t) \otimes \psi_{N}(t)^{\dagger}$.
	Moreover, $\Tr( H_{f,N} A_{\rho,N} )$ converges to the constant $\int f\rho $ as $N \to \infty$.
\end{theorem}
\begin{proof}
	As $H_{f,N}(t) = U_{N}(t) H_{f,N}(0) U_{N}(t)^{\dagger}$ and $\psi_{N}(t) = U_{N} \cdot \psi_{N}(0)$ we observe that
	\begin{align}
		\Tr( H_{f,N}(t) (\psi_{N} (t) \otimes \psi_{N}^{\dagger} (t) ) ) &= \Tr(  U_{N}(t) H_{f,N}(0) U_{N}(t)^{\dagger} U)_{N}(t) (\psi_{N}(0) \otimes \psi_{N}(0)^{\dagger}) U_{N}(t)^{\dagger}) \\
		&= \Tr( H_{f,N}(0) ( \psi_{N}(0) \otimes \psi_{N}(0)^{\dagger} ) )
	\end{align}
	Convergence follows from Theorems \ref{thm:function convergence} and \ref{thm:density convergence}.
\end{proof}

\section{Numerical Experiments} \label{sec:numerics}

This section describes two numerical experiments.  First, a benchmark computation to illustrate the spectral convergence of our method and the conservation properties in the case of a known solution is considered.

\subsection{Benchmark computation}
\label{sec:benchmark}
Consider the vector field $\dot{x} = -\sin(2 x)$ for $x \in S^{1}$.
The flow of this system is given by:
\begin{align}
	\Phi_{X}^{t}(x) = \operatorname{atan} \left( e^{2t} \tan( x) \right).
\end{align}
If the initial density is a uniform distribution, $\rho_{0} = dx$, then the the exact solution of \eqref{eq:density pde} is:
\begin{align}
	\rho(x;t) =  \left( e^{2t} \sin^{2}(x) + e^{-2t} \cos^{2}(x) \right)^{-1}  |dx| \label{eq:exact solution}
\end{align}
Figure \ref{fig:S1} depicts the evolution of $\rho(x;t)$ at $t=1.5$ with an initial condition.
Figure \ref{fig:exact} depicts the exact solution, given by \eqref{eq:exact solution},  Figure \ref{fig:standard spectral} depicts the numerical solution computed from a standard Fourier discretization of \eqref{eq:density pde} with 32 modes, and Figure \ref{fig:gn spectral} depicts the numerical solution solution computed using Algorithm \ref{alg:density} with 32 modes.

\begin{figure}[h!]
	\hspace*{-.5cm}
	\begin{subfigure}{0.36\textwidth}
		\includegraphics[width=0.9\textwidth]{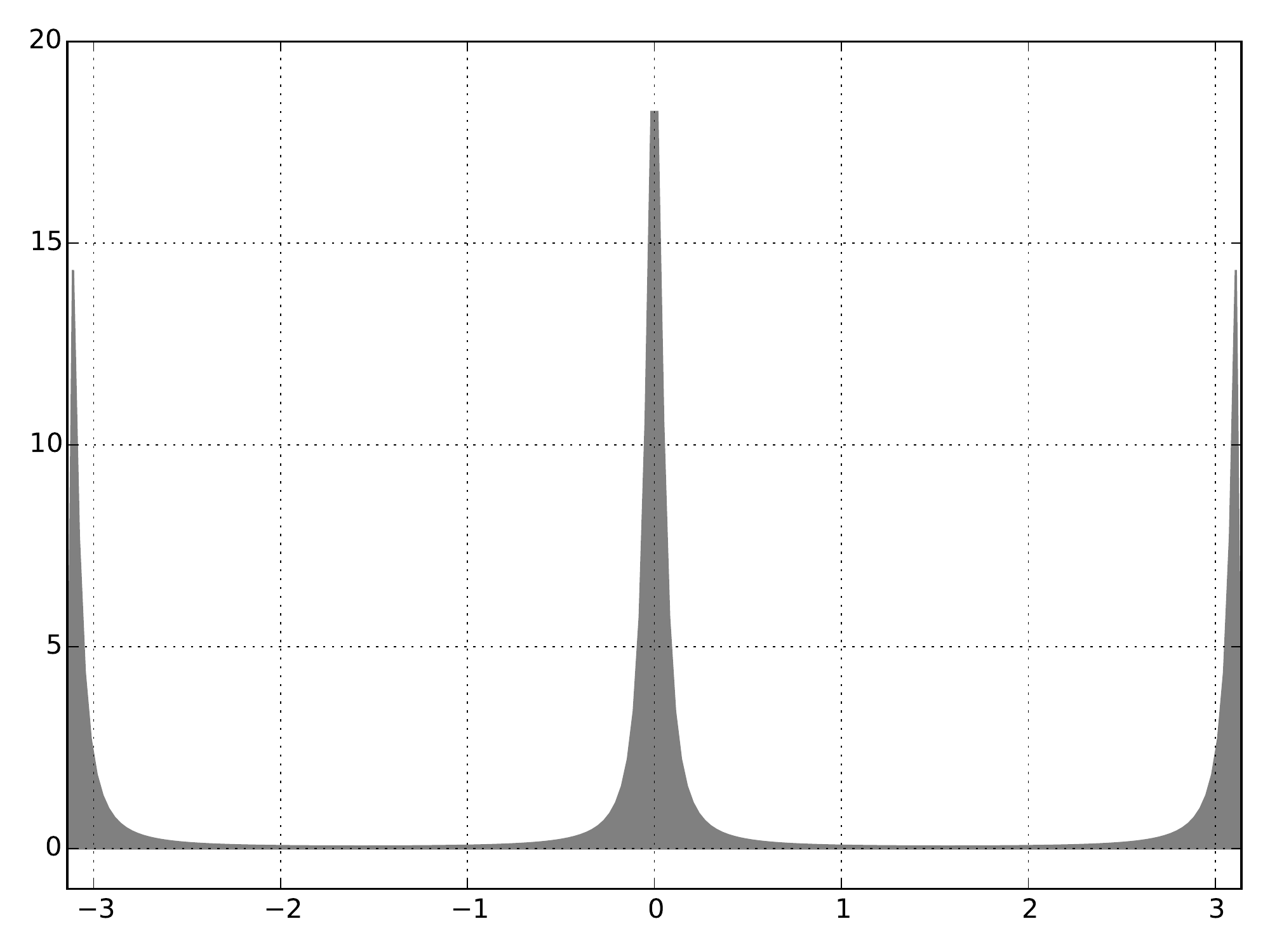}
		\caption{Exact}
		\label{fig:exact}
	\end{subfigure}
	\hspace*{-.65cm}
	\begin{subfigure}{0.36\textwidth}
		\includegraphics[width=0.9\textwidth]{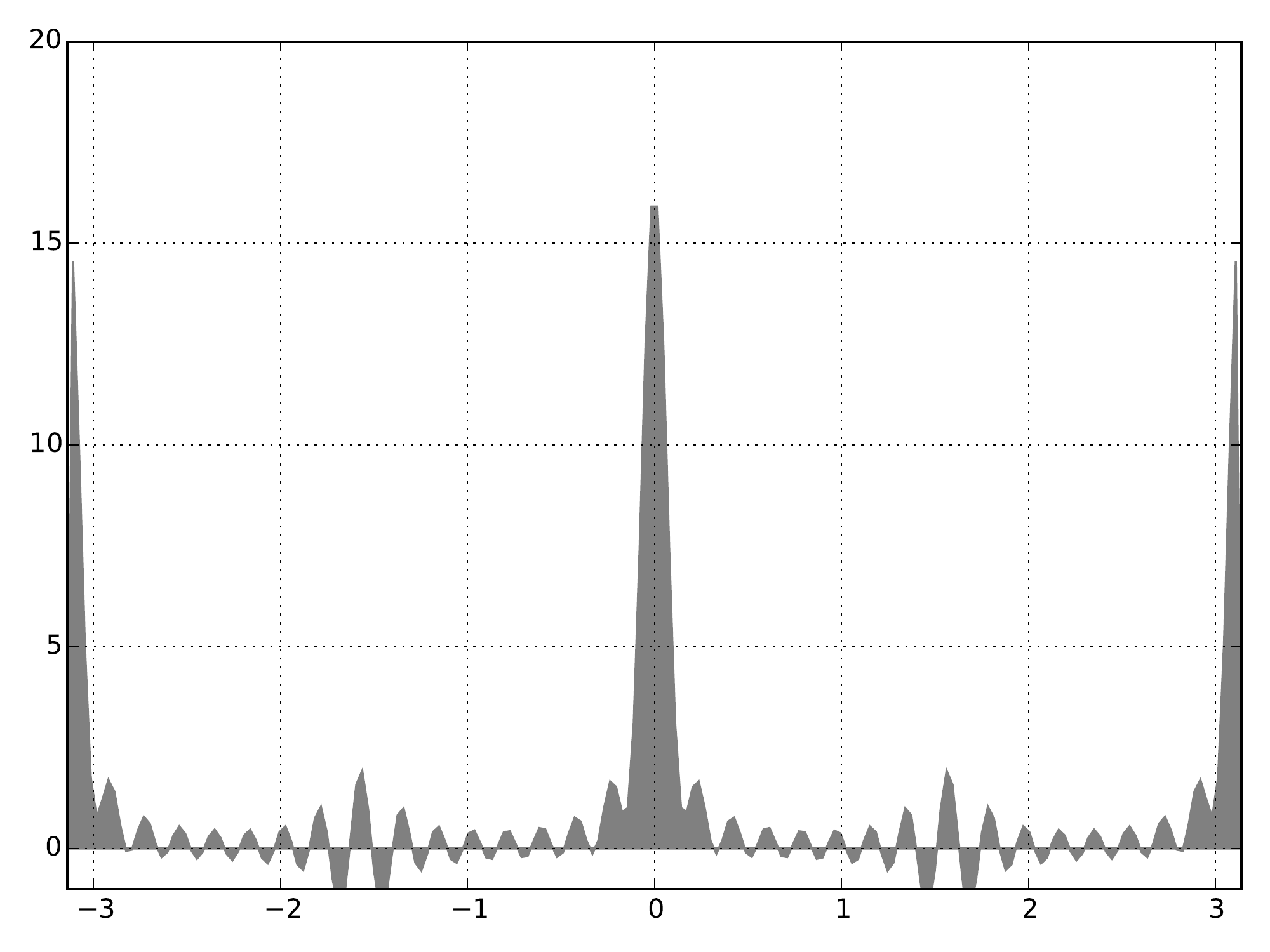}
		\caption{Standard spectral}
		\label{fig:standard spectral}
	\end{subfigure}
	\hspace*{-.65cm}
	\begin{subfigure}{0.36\textwidth}
		\includegraphics[width=0.9\textwidth]{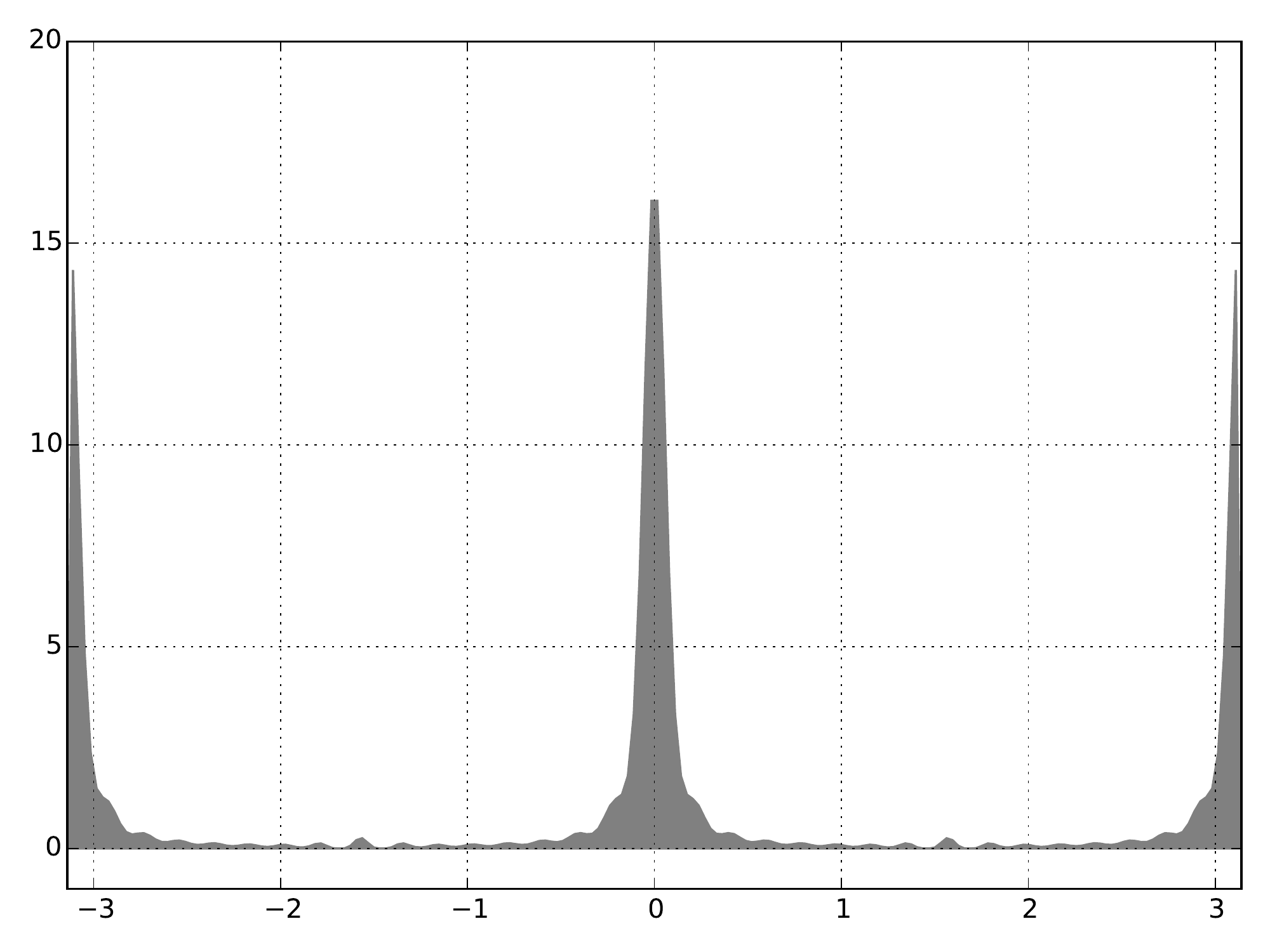}
		\caption{Algorithm \ref{alg:density}}
		\label{fig:gn spectral}
	\end{subfigure}
	\caption{A benchmark illustration of Algorithm \ref{alg:density} on the example described in Section \ref{sec:benchmark}.}
	\label{fig:S1}
\end{figure}

Here we witness how Algorithm \ref{alg:density} has greater qualitative accuracy than a standard spectral discretization, in the ``soft'' sense of qualitative accuracy.
For example, standard spectral discretization exhibits negative mass, which is not achievable in the exact system.
Moreover, the $L^{1}$-norm is not conserved in standard spectral discretization.  
In contrast, Theorem \ref{thm:norms} proves that the $L^{1}$-norm is conserved by Algorithm \ref{alg:density}.
A plot of the $L^{1}$-norm is given in Figure \ref{fig:L1}.
Finally, a convergence plot is depicted in Figure \ref{fig:convergence}.  
Note the spectral convergence of Algorithm \ref{alg:density}.
In terms of numerical accuracy, Algorithm \ref{alg:density} appears to have a lower coefficient of convergence.

\begin{figure}[h!]
	\hspace*{-1.2cm}
	\centering
	\includegraphics[width=0.8\textwidth]{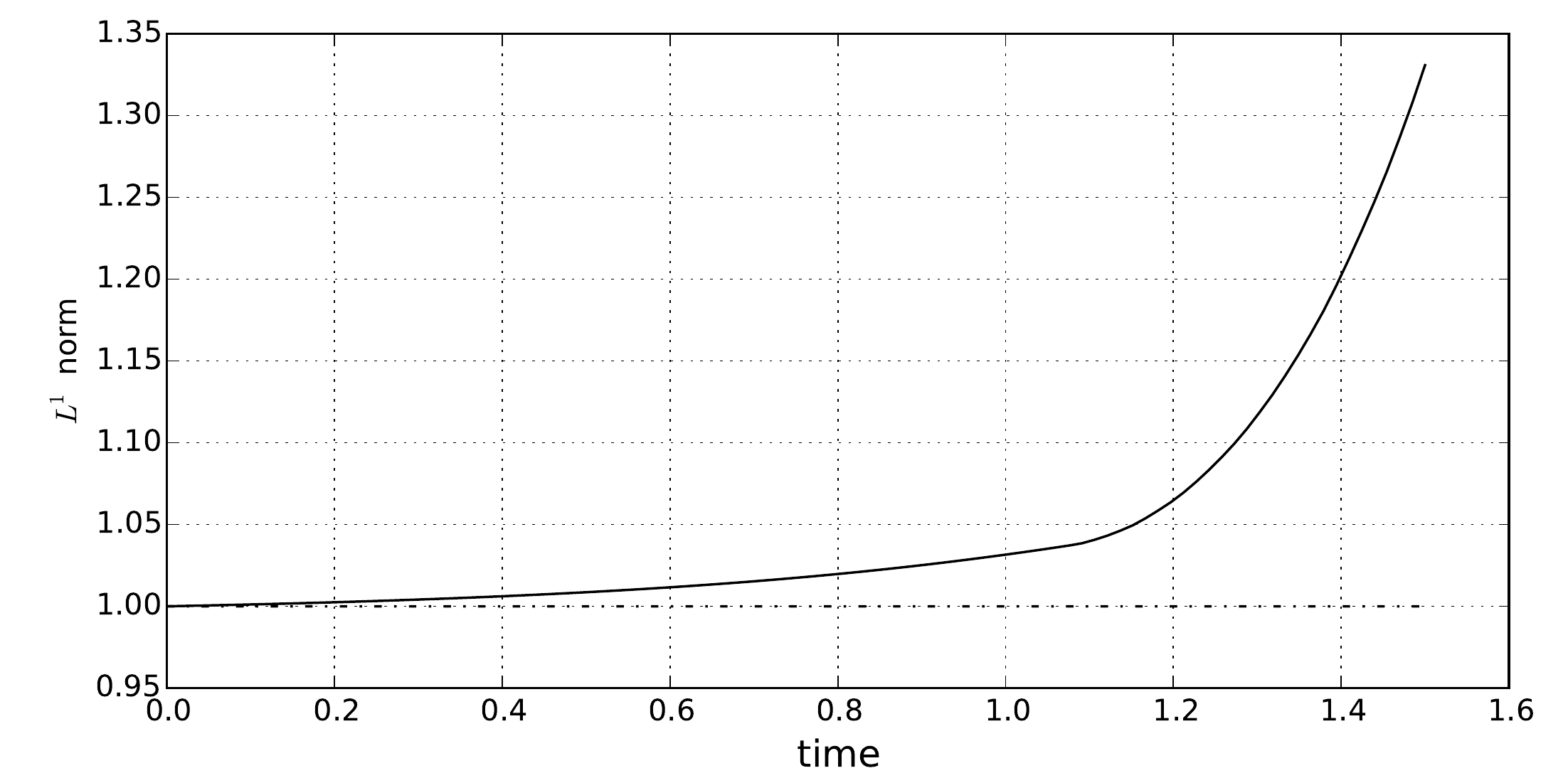}
	\caption{A plot of the $L^{1}$-norm vs time of a standard spectral discretization (solid) and the result of Algorithm \ref{alg:density} (dotted) on the example described in Section \ref{sec:benchmark}.}
	\label{fig:L1}
\end{figure}

\begin{figure}[h!]
	\hspace*{-1.2cm}
	\centering
	\includegraphics[width=0.9\textwidth]{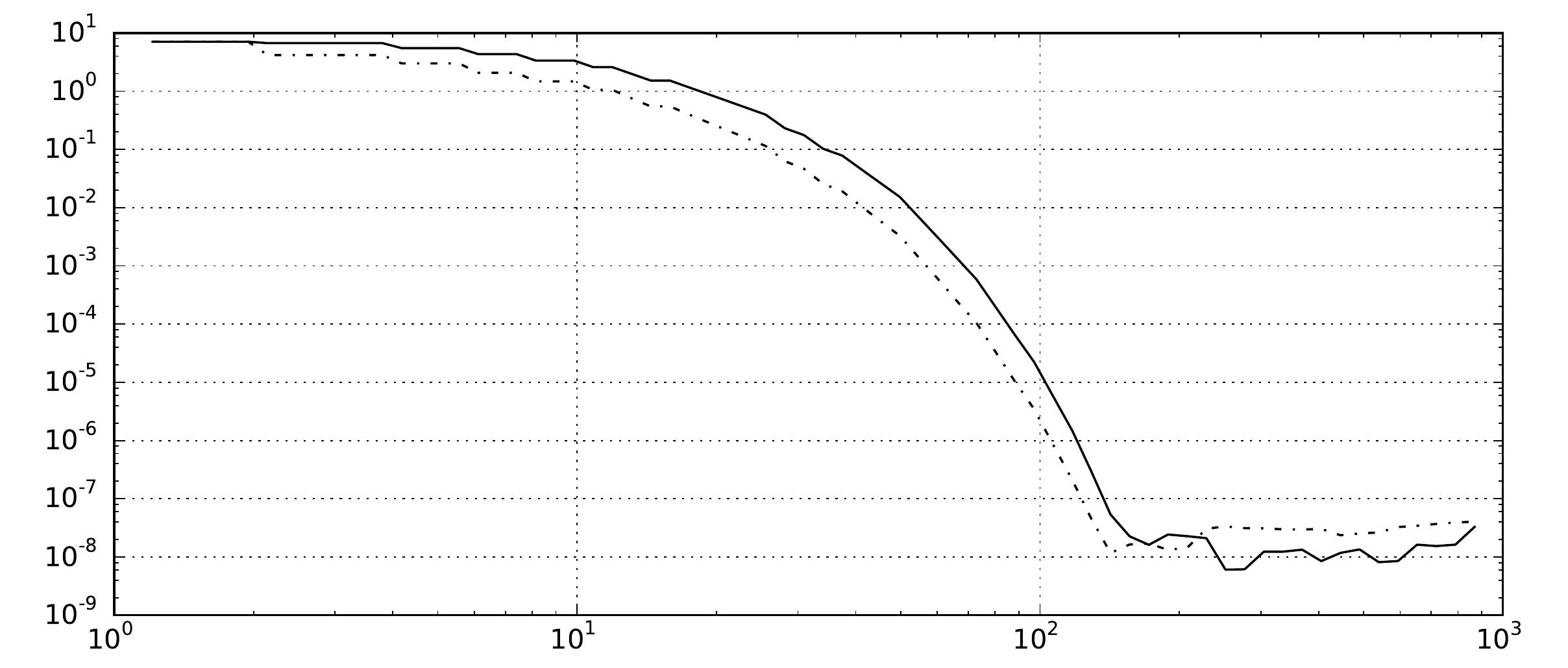}
	\caption{Convergence plot for Algorithm \ref{alg:density} (dotted) and a standard spectral method (solid) in the $L^{1}$-norm.}
	\label{fig:convergence}
\end{figure}

In general, Algorithm \ref{alg:function} is very difficult to work with, as it outputs an operator rather than a classical function.
However, Algorithm \ref{alg:function} is of theoretical value, in that it may inspire new ways of discretization (in particular, if one is only interested in a few level sets).
We do not investigate this potentiality here in the interest of focusing on the qualitative aspects of this discretization.
For example, under the initial conditions $g_{0}(x) = \cos(x)$ and $f_{0} = \sin(x)$ the exact solutions to \eqref{eq:quantum observable ode} are:
\begin{align*}
	g(x,t) &= \cos(x) \left( e^{4t} \sin^{2}(x) + \cos^{2}(x) \right)^{-1/2}\\
	f(x,t) &= \sin(x) \left( \sin^{2}(x) + e^{-4t} \cos^{2}(x) \right)^{-1/2}
\end{align*}
Under the initial condition $h_{0} = f_{0} \cdot g_{0}  = \sin(x) \cos(x)$ the exact solution to \eqref{eq:quantum observable ode} is:
\begin{align*}
	h(x,t) = f(x,t) g(x,t) = \cos(x) \sin(x) \left( \cos^{2}(x) + e^{4t} \sin^{2}(x) \right)^{-1}.
\end{align*}
One can compute $h$ by first multiplying the initial conditions and then using Algorithm \ref{alg:function} to evolve in time, or we may evolve each initial condition in time first, and multiply the outputs.
If one uses Algorithm \ref{alg:function}, then both options, as a result of Theorem \ref{thm:algebra}, yield the same result up to time discretization error (which is obtained with error tolerance $1e-8$ in our code).
In contrast, if one uses a standard spectral discretization, then these options yield different results with a discrepancy.
This discrepancy between the order of operations for both discretization methods is depicted in Figure \ref{fig:discrepancy}.

Finally, the sup-norm is preserved by the solution of \eqref{eq:function pde}.
As shown in Theorem \ref{thm:quantize}, the sup-norm is equivalent to the operator norm when the functions are represented as operators on $L^{2}(M)$.
As proven by Theorem \ref{thm:norms}, the operator-norm is conserved by Algorithm \ref{alg:function}.
In contrast, the sup-norm drifts over time under a standard discretization.  
This is depicted in Figure \ref{fig:norms}

\begin{figure}[h!]
	\hspace*{-1.2cm}
	\includegraphics[width=1.15\textwidth]{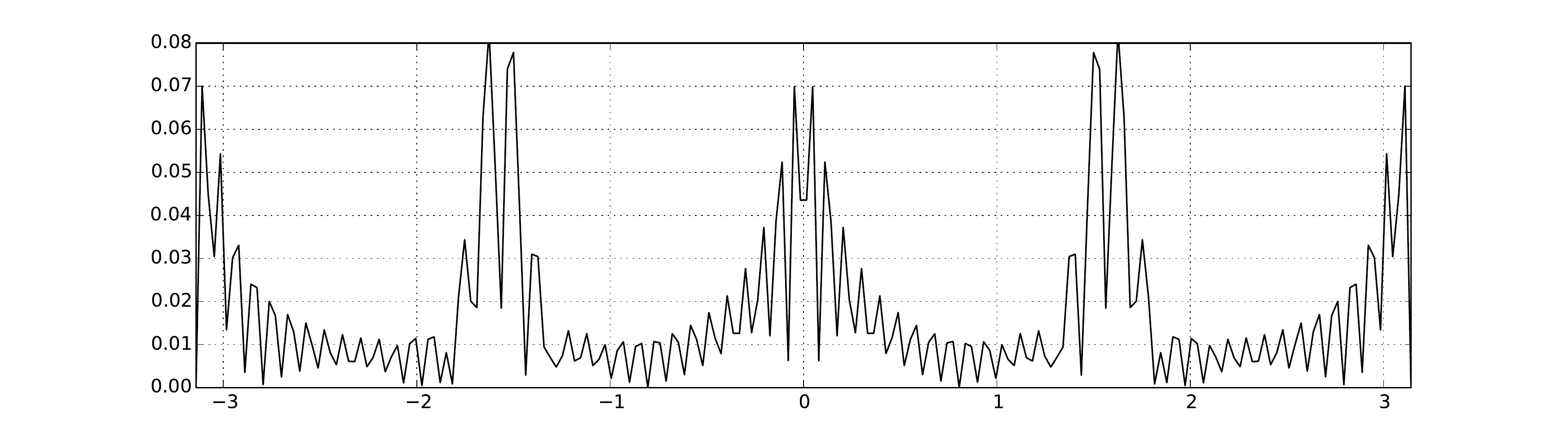}
	\caption{The discrepancy due to non-preservation of scalar products under a standard spectral Galerkin discretization. 
	The discrepancy of Algorithm \ref{alg:function} (not plotted) is attributable to our time-discretization scheme where we only tolerated error of $10^{-8}$ in this instance.}
	\label{fig:discrepancy}
\end{figure}  

\begin{figure}[h!]
	\hspace*{-1.2cm}
	\centering
	\includegraphics[width=0.8\textwidth]{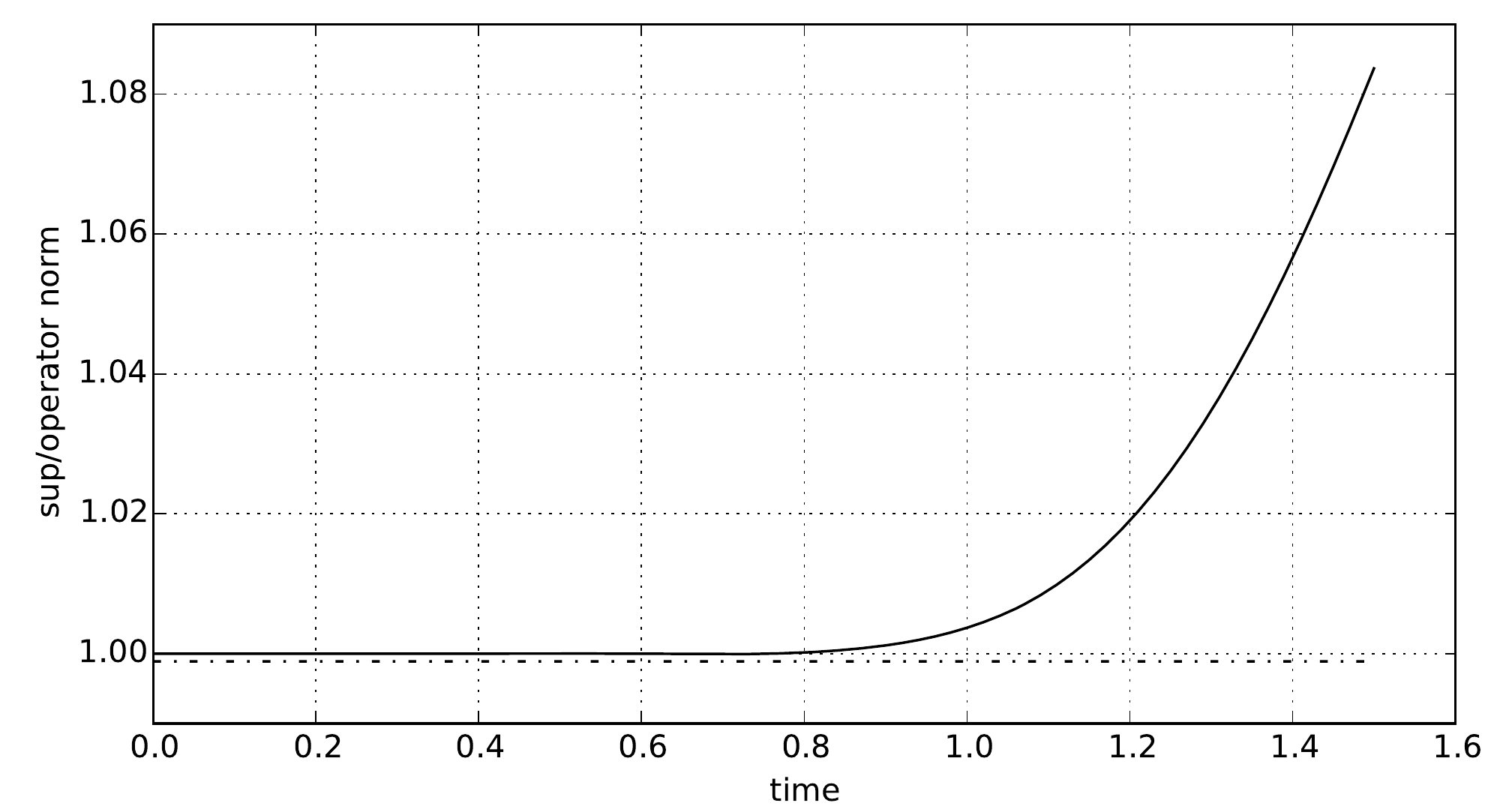}
	\caption{A plot of the sup-norm vs time of a standard spectral discretization (blue) and the result of Algorithm \ref{alg:function} (red) on the example described in Section \ref{sec:benchmark}.}
	\label{fig:norms}
\end{figure}

\subsection{A modified ABC flow}
\label{sec:ABC_flow}

Consider the system
\begin{align}
	\dot{x} = A \sin( 2\pi z) + C \cos(2\pi y)  + D \cos(2\pi x)\\
	\dot{y} = B \sin( 2\pi z) + A \cos(2\pi y)  + D \cos(2\pi y)\\
	\dot{z} = A \sin( 2\pi z) + B \cos(2\pi y)  + D \cos(2\pi z)
\end{align}
on the three-torus for constants $A,B,C,D \in \mathbb{R}$.  
When $D=0$ this system is the well studied volume conserving system known as an Arnold-Beltrami-Childress flow~\cite{ArnoldKhesin1992}.
When $A > B > C > 0$, $D=0$, and $C << 1$, then the solutions to this ODE are chaotic, with a uniform steady state distribution ~\cite{MajdaBertozzi2002}.
When $D=0$ the operator $\pounds_{X}$ of \eqref{eq:half density pde} is identical to the operator $\partial_{\alpha}( \rho X^{\alpha})$ that appears in \eqref{eq:density pde}, and Algorithm \ref{alg:half density} do not differ from a standard spectral discretization.\footnote{This is always the case for a volume conserving system.}
Therefore we consider the case where $D > 0$ to see how our discretization are differs from the standard one.
When $D> 0$ volume is no longer conserved and there is a non-uniform steady-state distribution.

For the following numerical experiment let $A=1.0,B=0.5,C=0.2,$ and $D=0.5$.
As an initial condition consider a wrapped Gaussian distribution with anisotropic variance $\sigma= (0.2, 0.3, 0.3)$ centered at $(0,0,0)$.
Equation \eqref{eq:density pde} is approximately solved using Algorithm \ref{alg:density}, Monte-Carlo, and a standard spectral method.
The results of the $z$-marginal of these densities are illustrated in Figure \ref{fig:ABCD}.
The top row depicts the results from using Algorithm \ref{alg:density} using $33$ modes along each dimension.
The middle row depicts the results from using a Monte-Carlo method with $15^{3} = 3375$ particles as a benchmark computation.
Finally, the bottom row depicts the results from using a standard Fourier based discretization of \eqref{eq:density pde} using 33 modes along each dimension.
Notice that Algorithm \ref{alg:density} performs well when compared to the standard discretization approach.

\begin{figure}[h!]
	\centering
	\includegraphics[width=1\textwidth]{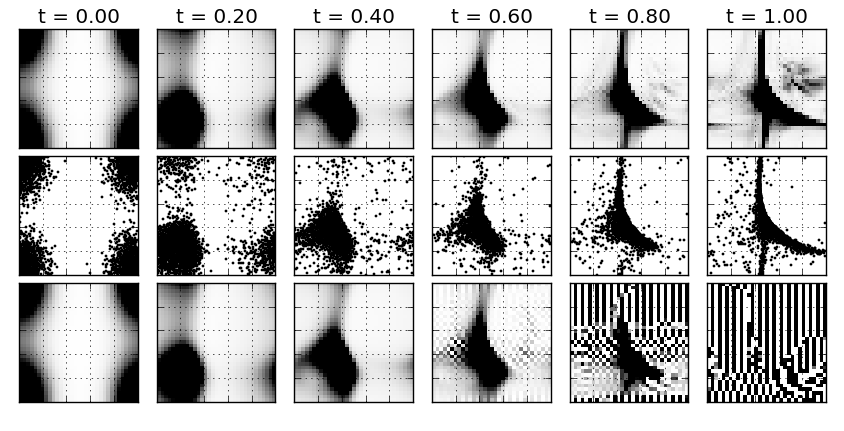}
	\caption{An illustration of the performance of Algorithm \ref{alg:density} (top row), Monte Carlo (middle row) and a standard spectral Galerkin (bottom row) on the example described in Section \ref{sec:ABC_flow}.
	The domain is the $2$-torus.  Here we've consider an initial probability density given by a wrapped Gaussian. Darker regions represent areas of higher-density.}
	\label{fig:ABCD}
\end{figure}

\section{Conclusion}

In this paper we constructed a numerical scheme for \eqref{eq:function pde} and \eqref{eq:density pde} that is spectrally convergent and qualitatively accurate, in the sense that natural invariants are preserved.
The result of obeying such conservation laws is a robustly well-behaved numerical scheme at a variety of resolutions where legacy spectral methods fail.
This claim was verified in a series of numerical experiments which directly compared our algorithms with standard Fourier spectral algorithms.
The importance of these conservation laws was addressed in a short discussion on the Gelfand Transform.
We found that conservation laws completely characterize \eqref{eq:function pde} and \eqref{eq:density pde}, and this explains the benefits of using qualitatively accurate scheme at a more fundamental level.

\subsection{Acknowledgements}
This paper developed over the course of years from discussions with many people whom we would like to thank: Jaap Eldering, Gary Froyland,
 	Darryl Holm, Peter Koltai, Stephen Marsland, Igor Mezic, Peter Michor, Dmitry Pavlov, Tilak Ratnanather, and Stefan Sommer. 
This research was made possible by funding from the University of Michigan.

\end{document}